\documentclass[AMA,STIX1COL]{WileyNJD-v2}

\pdfoutput=1

\articletype{Research Article}%

\received{26 April 2016}
\revised{6 June 2016}
\accepted{6 June 2016}

\usepackage{bm}
\usepackage{units}
\usepackage{xspace}
\usepackage[
  capitalize,           
]{cleveref}             

\graphicspath{{figs/}}

\newcommand*{\red}[1]{\textcolor{red}{#1}}        

\newcommand*{\ie}{i.e.,\xspace}
\newcommand*{\cf}{,~see\xspace}
\newcommand*{\eg}{e.g.\xspace~}

\newcommand*{\CAL}[1]{\ensuremath{\mathcal{#1}}}
\newcommand*{\abs}[1]{\lvert#1\rvert}	 
\newcommand*{\norm}[1]{\lVert#1\rVert}  
\newcommand*{\Matrix}[1]{\begin{bmatrix} #1 \end{bmatrix}}
\newcommand*{\MAT}[1]{\ensuremath{\bm{#1}}}
\newcommand*{\VEC}[1]{\MAT{#1}}
\newcommand*{\E}{\mathop{}\!\mathrm{E}}			
\newcommand*{\T}{^\textup{\textsf{T}}}			
\newcommand*{\inv}{^{-1}}

\newcommand*{\R}{\ensuremath{\mathbb{R}}}
\newcommand*{\N}{\ensuremath{\mathbb{N}}}
\newcommand*{\K}{\ensuremath{\CAL K}\xspace}
\newcommand*{\Kinfty}{\ensuremath{\CAL K_\infty}\xspace}
\newcommand*{\KL}{\ensuremath{\CAL{KL}}\xspace}
\newcommand*{\GP}{Gaussian process\xspace}
\newcommand*{\GPs}{Gaussian processes\xspace}
\newcommand*{\MPC}{model predictive control\xspace}
\newcommand*{\OCP}{optimal control problem\xspace}
\newcommand*{\uc}{uniformly continuous\xspace}

\newcommand*{\D}{\ensuremath{\CAL D}\xspace}
\newcommand*{\Dk}{\ensuremath{\CAL D_k}\xspace}
\newcommand*{\Dkp}{\ensuremath{\CAL D_{k+1}}\xspace}
\newcommand*{\Dskp}{\ensuremath{\CAL D^{'}_{k+1}}\xspace}
\newcommand*{\DO}{\ensuremath{\CAL D_0}\xspace}
\newcommand*{\Dref}{\ensuremath{\CAL D_\text{ref}}\xspace}
\newcommand*{\Draw}{\ensuremath{\CAL D_\text{raw}}\xspace}
\newcommand*{\Dcomb}{\ensuremath{\CAL D_\text{comb}}\xspace}
\newcommand*{\Op}{\ensuremath{\CAL O}\xspace}
\newcommand*{\OO}{\ensuremath{\Omega_r(\lambda)}\xspace}
\newcommand*{\OOo}{\ensuremath{\Omega^0_r(\lambda)}\xspace}
\newcommand*{\OOk}{\ensuremath{\Omega^k_r(\lambda)}\xspace}
\newcommand*{\U}{\ensuremath{\CAL U}\xspace}
\newcommand*{\Y}{\ensuremath{\CAL Y}\xspace}
\newcommand*{\Ys}{\ensuremath{\CAL Y_\text{s}}\xspace}
\newcommand*{\Yh}{\ensuremath{\CAL Y_\text{h}}\xspace}
\newcommand*{\X}{\ensuremath{\CAL X}\xspace}
\newcommand*{\Xf}{\ensuremath{\CAL X_\text{f}}\xspace}
\newcommand*{\XN}{\ensuremath{\CAL X_N(\lambda)}\xspace}
\newcommand*{\XNo}{\ensuremath{\CAL X^0_N(\lambda)}\xspace}
\newcommand*{\XNk}{\ensuremath{\CAL X^k_N(\lambda)}\xspace}

\newcommand*{\Nx}{\ensuremath{{n_x}}\xspace}
\newcommand*{\Nu}{\ensuremath{{n_u}}\xspace}
\newcommand*{\Nw}{\ensuremath{{n_w}}\xspace}

\newcommand*{\Nsim}{\ensuremath{N_\textup{sim}}\xspace}
\newcommand*{\Nstep}{\ensuremath{N_\textup{step}}\xspace}

\newcommand*{\A}{\MAT A\xspace}

\newcommand*{\G}{\MAT G\xspace}
\renewcommand*{\P}{\MAT P\xspace}
\newcommand*{\KK}{\MAT K\xspace}
\newcommand*{\KKi}{\ensuremath{\MAT K^{-1}}\xspace}
\newcommand*{\Q}{\MAT Q\xspace}
\newcommand*{\RR}{\MAT R\xspace}
\renewcommand*{\S}{\MAT S\xspace}

\newcommand*{\My}{\ensuremath{m_y}\xspace}
\newcommand*{\Mu}{\ensuremath{m_u}\xspace}
\newcommand*{\Sn}{\ensuremath{\sigma^2_\text n}\xspace}
\newcommand*{\Sf}{\ensuremath{\sigma^2_\text f}\xspace}
\newcommand*{\sigmaU}{\ensuremath{\bar\sigma^2}\xspace}
\newcommand*{\Vf}{\ensuremath{V_\text{f}}\xspace}
\newcommand*{\VN}{\ensuremath{V_N}\xspace}
\newcommand*{\cov}{\ensuremath{k}\xspace} 
\newcommand*{\kf}{\ensuremath{\kappa_\text{f}}\xspace}
\newcommand*{\kMPC}{\ensuremath{\kappa_\text{MPC}}\xspace}
\newcommand*{\ls}{\ensuremath{\ell_\text s}\xspace}
\newcommand*{\lb}{\ensuremath{\ell_\text b}\xspace}

\renewcommand*{\b}{\ensuremath{\VEC b}\xspace}
\renewcommand*{\k}{\ensuremath{\VEC k}\xspace}
\newcommand*{\q}{\ensuremath{\VEC q}\xspace}
\newcommand*{\e}{\ensuremath{\VEC e}\xspace}

\newcommand*{\ekV}{\ensuremath{\VEC e_k}\xspace}

\newcommand*{\ep}{\ensuremath{e^\text{p}}\xspace}
\newcommand*{\epU}{\ensuremath{\bar e}\xspace}
\newcommand*{\uk}{\ensuremath{u_k}\xspace}
\newcommand*{\s}{\VEC s\xspace}
\newcommand*{\w}{\VEC w\xspace}
\newcommand*{\wk}{\ensuremath{\VEC w_k}\xspace}
\newcommand*{\x}{\VEC x\xspace}
\newcommand*{\xk}{\ensuremath{\VEC x_k}\xspace}
\newcommand*{\y}{\VEC y\xspace}
\newcommand*{\z}{\VEC z\xspace}
\newcommand*{\zk}{\VEC z_k\xspace}
\newcommand*{\ykp}{\ensuremath{y_{k+1}}\xspace}
\newcommand*{\uref}{\ensuremath{u_\text{ref} }\xspace}
\newcommand*{\xref}{\ensuremath{\VEC x_\text{ref} }\xspace}
\newcommand*{\yref}{\ensuremath{y_\text{ref}}\xspace}

\newcommand*{\hx}{\ensuremath{\hat{\VEC x}}\xspace}
\newcommand*{\hu}{\ensuremath{\hat u}\xspace}

\newcommand*{\hxk}{\ensuremath{\hat{\VEC x}_k}\xspace}
\newcommand*{\hxkp}{\ensuremath{\hat{\VEC x}_{k+1}}\xspace}
\newcommand*{\hyk}{\ensuremath{\hat y_k}\xspace}
\newcommand*{\hykp}{\ensuremath{\hat y_{k+1}}\xspace}

\newcommand*{\fN}{\ensuremath{\hat f}\xspace}
\newcommand*{\FN}{\ensuremath{\hat F}\xspace}

\newcommand*{\xE}{\ensuremath{\tilde{\VEC x}}\xspace}
\newcommand*{\xEp}{\ensuremath{\tilde{\VEC x}^+}\xspace}
\newcommand*{\yk}{\ensuremath{y_k}\xspace}
\newcommand*{\uF}{\ensuremath{\mathbf{\hat u}}\xspace}
\newcommand*{\xF}{\ensuremath{\hat{\mathbf x}}\xspace}
\newcommand*{\uS}{\ensuremath{\hat{\mathbf u}_{k|k}}\xspace}
\newcommand*{\uSp}{\ensuremath{\hat{\mathbf u}_{k+1|k}}\xspace}
\newcommand*{\xkp}{\ensuremath{\VEC x_{k+1}}\xspace}

\newcommand*{\alphaa}{\VEC \alpha\xspace}
\newcommand*{\betaa}{\VEC \beta\xspace}
\newcommand*{\thetaa}{\VEC \theta\xspace}

\newcommand*{\wb}{\ensuremath{\mathbf w}\xspace}
\newcommand*{\zb}{\ensuremath{\mathbf z}\xspace}
\newcommand*{\postM}{\ensuremath{m_+}\xspace}
\newcommand*{\postV}{\ensuremath{\sigma^2_+}\xspace}

\begin{document}

\title{Online learning-based Model Predictive Control with\\
Gaussian Process Models and Stability Guarantees
}

\author[1]{Michael Maiworm}

\author[2]{Daniel Limon}

\author[1]{Rolf Findeisen*}

\authormark{Maiworm \textsc{et al}}

\address[1]{\orgdiv{Laboratory for Systems Theory and Automatic Control},
\orgname{Otto-von-Guericke-University}, \orgaddress{\city{Magdeburg},
\country{Germany}}}

\address[2]{\orgdiv{Department of Systems Engineering and Automation},
\orgname{Universidad de Sevilla}, \orgaddress{\city{Seville},
\country{Spain}}}

\corres{*R. Findeisen, Institute for Automation Engineering,
Otto-von-Guericke-University, Magdeburg,
Germany. \email{rolf.findeisen@ovgu.de}}


\abstract[Summary]{
  Model predictive control allows to provide high performance and safety
  guarantees in the form of constraint satisfaction.
  These properties, however, can be satisfied only if the underlying
  model, used for prediction, of the controlled process is sufficiently
  accurate.
  One way to address this challenge is by data-driven and machine
  learning approaches, such as \GPs, that allow to refine the model
  online during operation.
  We present a combination of an output feedback
  \MPC scheme and a \GP-based prediction model that is capable of
  efficient online learning.
  To this end, the concept of evolving \GPs is combined with recursive
  posterior prediction updates.
  The presented approach guarantees recursive constraint satisfaction and
  input-to-state stability with respect to the model-plant mismatch.
  Simulation studies underline that the \GP prediction model can be
  successfully and efficiently learned online.
  The resulting computational load is significantly reduced via the
  combination of the recursive update procedure and by limiting the
  number of training data points while maintaining good performance.
}


\jnlcitation{\cname{%
\author{Maiworm M.}, 
\author{D. Limon}, and
\author{R. Findeisen}} (\cyear{2020}), 
\ctitle{Online learning-based Model Predictive Control with Gaussian
Process Models and Stability Guarantees}, \cjournal{International
Journal of Robust and Nonlinear Control}, \cvol{\red{2020;00:1--6}}.}

\maketitle


\section{Introduction}
\label{sec:introduction}

Model predictive control (MPC) \cite{Rawlings2009} is naturally capable
of dealing with multi-input multi-output systems and constraints on the
input, state, and output already in the design process.
This has led to manifold scientific interest, as well as practical
applications. \cite{Mayne2014, Lucia2016}
In terms of performance, MPC can be superior to other control approaches
because the prediction of the process under consideration allows to
compute control actions based on future outcomes and facilitates to take
preview information about references and disturbances into account.
Hence, the prediction model plays a crucial role in MPC.
Unfortunately, there is always a certain process-model error or model
uncertainty present in practice and the system might change over time,
which limits the prediction quality of the model.
One way to deal with this situation is to resort to robust MPC schemes,
such as, for instance, min-max MPC \cite{Scokaert1998}, tube-based MPC
\cite{Mayne2006}, multi-scenario approaches \cite{Lucia2013,
Maiworm2015}, or stochastic approaches\cite{Paulson2018} that take the
uncertainty explicitly into account.

Prediction models are often based on first principles approaches,
which can be very time consuming or even impossible in practice.
Furthermore, if the underlying process or environmental conditions
change, a once good model can degrade and thus needs to be adapted.
An alternative to first principles approaches is to derive
prediction models directly from measured data.
The resulting models, so-called black or grey box models
\cite{Ljung1999}, can in principle be learned or refined during
operation by including newly available data.
Thereby, they can account for changing process dynamics or a
changing process environment.
Combining data-driven with first principles models is another
possibility.\cite{Yang2015, Ostafew2016, Hewing2018}

Although data-driven modeling is not a new field of research, it gained
significant attention over the last years due to increasing
computational power, the possibility to widely collect data, and the
rise of machine learning algorithms, such as neural networks, deep
learning, support vector machines, or Gaussian processes
(GPs)\cite{Rasmussen2006, Kocijan2016}.
Especially the use of GPs within MPC
has attracted significant interest in recent years\cite{Kocijan2003,
Yang2015, Klenske2016, Ostafew2016, Cao2017, Maiworm2018b}. 
However, combining GPs with MPC leads to multiple challenges, such as
the cubical increase of the computational load with the number of
training data points.
This also increases the overall necessary computations to solve
the resulting \OCP.
Furthermore, the utilization of GPs in an \OCP can render the resulting
optimization very nonlinear, even for a small number of data points,
which increases the probability of obtaining suboptimal or infeasible
solutions.
Despite these challenges, GPs are employed together with MPC as they
provide several advantages.
For instance, they do not only allow to compute a prediction
of the system evolution but also a prediction variance (an effective
measure of the uncertainty of the learned model), they are less
susceptible to overfitting, and they have, under certain
circumstances, universal approximation capabilities for a large class of
functions\cite{Steinwart2008}, thereby allowing to model
the underlying dynamics of a wide variety of systems.

In order to reduce the computational load of GPs one can distinguish two
main approaches.
The first approach basically fixes the maximum number of training
data points, while the second approach employs so-called
sparsity\cite{Snelson2006, Lazaro2010}.
The first approach often entails the drawback that the GP might
not be able to model the system with sufficient accuracy throughout the
full operation space.
To compensate for this, one can resort to online learning (or
adaptation) of the \GP during operation, which also allows to account
for time-varying systems or changing environmental conditions.
On the downside, some of the computation time that is saved
by reducing the number of training data points is in turn spent by the
learning process, which includes updates of the training data set and
covariance matrix, recalculation of the covariance matrix inverse, and
hyperparameter optimization in each time step.
While these often computationally extensive calculations can be
performed offline, only very few publications exist that combine MPC
with online learning of GPs.
The required computations often take too long to control most processes.
Thus, GPs are mostly trained/learned offline. \cite{Ostafew2014,
Cao2017, Hewing2017}
Exceptions are, for instance, the works by Ortman et
al\cite{Ortmann2017}, where the system had a large time constant in the
order of hours or Klenske et al\cite{Klenske2016}, which provided a
hyperparameter optimization tailored to the specific application.

Another important aspect when combining \GPs and \MPC is safety,
constraint satisfaction, and stability, for which different approaches
have been proposed.
One can, for example, avoid to enforce stability by design and include
instead the GP posterior variance in the cost function of the \OCP.
This avoids steering the plant into regions where the model validity is
questionable.\cite{Murray2003, Kocijan2005b, Azman2008}
Also, one can perform a posteriori stability verification.
For instance, Berkenkamp et al\cite{Berkenkamp2016a} proposed to learn
the region of attraction of a given closed-loop system, whereas
Vinogradska et al\cite{Vinogradska2016} calculated invariant sets for the
validation of stability in a closed-loop with GP models.
Another possibility is to use invariant \emph{safe sets} and employ a
two-layer control framework, where a safe controller is combined with a
control policy that optimizes performance.\cite{Akametalu2014,
Koller2018, Fisac2018, Wabersich2018}
For instance, in the works by Aswani et al\cite{Aswani2013} and Bethge
et al\cite{Bethge2018} two different prediction models were used in
parallel, where the first is a nominal model, used to guarantee robust
stability using tubes, and the other can be a general learning-based
model (\eg a \GP) used to optimize performance.
In the work of Soloperto et al\cite{Soloperto2018} tube-based MPC was
considered together with GPs, which were also used to derive
robust stability.
To this end, uncertainty sets that are based on the GP variance were used
to construct tightened state and input constraint sets.
Since the uncertainty sets hold probabilistically, the same goes for the
stability result.
The two-layer framework was extended to three layers in
Bastani\cite{Bastani2019}.
The aforementioned approaches are based on the assumption of full state
information and the use of invariant terminal regions.

In Maiworm et al\cite{Maiworm2018b} we considered an output nominal MPC
scheme (which does not require full state information nor terminal
region in the \OCP) with an offline trained GP prediction model and
combined it with input-to-state stability (ISS), a framework that
covers inherent robust stability of nominal MPC and stability of
robust MPC schemes in the presence of constraints\cite{Limon2009}.
If a system under a predictive controller is shown to be ISS, then
this property is preserved even in the case of suboptimal solutions of
the involved optimal control problem.
We outlined conditions under which the GP-MPC scheme is inherently
robustly stable (\ie bounded disturbances lead to bounded effects on the
output) and guarantees recursive constraint satisfaction.
To this end, the uncertainty or disturbance has to be bounded
deterministically.
At the expense of a potentially smaller domain of attraction, the
advantage of guaranteeing inherent robust stability lies in its
simplicity.
The already involved ingredients in MPC merely have to satisfy certain
properties (e.g. uniform continuity).
The aforementioned methods in the literature on the other hand are
conceptually more complex and/or more computationally expensive than the
nominal MPC case because different control layers with backup
controllers are required \cite{Akametalu2014, Koller2018, Fisac2018,
Wabersich2018, Bastani2019}, different prediction models are employed
that have to be evaluated in parallel \cite{Bethge2018}, or tubes have
to be computed \cite{Soloperto2018}.
Furthermore, since the employed MPC formulation provides guarantees
without a terminal region, then if also no state constraints have
to be fulfilled, the resulting optimal control problem is easier to
solve.

In this work, we extend our previous results to the case of a limited
training data set of the \GP and aim towards online learning for a wide
class of applications.
To reduce the computational load we do not consider online
hyperparameter optimization.
Instead, we focus on a recursive approach to adapt the training data set
and compute the inverse covariance matrix tailored to MPC.
The main contributions of this work are:
\begin{itemize}
  \item Online learning of the GP model, by means of adaptation of the
    training data set, at reduced computational cost.
    This facilitates the possibility of deployment for faster processes.
    For this purpose, we employ a recursive formulation to update the GP
    prediction model online.
  \item Guaranteed input-to-state stability with constraint satisfaction
    for the presented online learning approach.
    The result is not confined to \GPs but holds for general
    prediction models that are learned online and satisfy the presented
    conditions.
  \item The extension of the method such that it yields good performance
    with only limited prior process knowledge (\eg lack of training data
    in important regions of the operation space).
    To this end, we incorporate the concept of evolving GPs to facilitate
    online learning by means of adaptation of the training data set.
    \cite{Petelin2011, Kocijan2016} 
    We derive criteria that use the GP prediction error and the variance
    to determine which points to add to the training data set.
  \item The use of analytic linearized GP models for the determination
    of the MPC terminal components.
\end{itemize}

The paper is structured as follows:
The considered problem setup is formulated in
Section~\ref{sec:problem_formulation}.
The concept of \GPs, together with the recursive formulation for online
learning, is outlined in Section~\ref{sec:gaussian_processes}
and used for the formulation of the \OCP in
Section~\ref{sec:model_predictive_control}.
The same section also contains the stability results.
Section~\ref{sec:simulations} presents simulation results with
focus on online learning of the \GP before
Section~\ref{sec:conclusion} concludes the paper.

~\\
\emph{Notation} \quad
Vectors, matrices, and sequences (of vectors or scalars) are set using
bold variables.
For matrices we use upper case ($\MAT Y$), for vectors
slanted lower case ($\VEC y$), and for sequences upright lower
case ($\mathbf{y}$).
Sets are denoted by calligraphic upper case variables ($\CAL Y$).
The distance of a point $\z \in \R^p$ to a set $\Y \subset \R^p$ is
defined as $d(\z,\Y) = \inf_{\y \in \Y} \norm{\y-\z}_\infty$, where
$\norm{\cdot}_\infty$ is the infinity norm (\ie $d(\z,\Y) = 0$ if $\z
\in \Y$).
If not stated otherwise, $\norm{\cdot}$ denotes the Euclidean vector
norm.
A function $\alpha : \R_{\geq 0} \to \R_{\geq 0}$ is a \K-function
if it is continuous, $\alpha(0) = 0$, and if it is strictly increasing.
A function $\alpha : \R_{\geq 0} \to \R_{\geq 0}$ is a \Kinfty-function
if it is a \K-function and unbounded.
A function $\beta : \R_{\geq 0} \times \R_{\geq 0} \to \R_{\geq 0}$ is a
\KL-function if $\beta(s,t)$ is \Kinfty in $s$ for any value of
$t$ and $\lim_{t \to \infty} \beta(s,t) = 0, \forall s \geq 0$.

\section{Problem Formulation}
\label{sec:problem_formulation}

We consider nonlinear discrete-time systems 
represented by a nonlinear auto\-regressive model with exogenous input
(NARX)\footnote{
  Under certain observability assumptions\cite{Levin1997}, a NARX model is
  sufficient to describe the dynamics of a wide class of systems.
}
\begin{subequations}
\begin{align}
  \ykp = ~& f(\xk, \uk) + \epsilon  \label{eq:NARX} \\
  \text{s.t.~} & \uk \in \U  \\
               & \yk \in \Y \ .
\end{align}
\label{eq:NARXsystem}
\end{subequations}
Here $k$ denotes the discrete time index, $\uk \in \R$ the
input, $\yk \in \R$ the output, and $\xk \in \R^\Nx$ is the NARX
``state vector''
\begin{align}
  \xk = \big[ \yk \ \cdots \ y_{k-\My} \ u_{k-1} \ \cdots \ u_{k-\Mu}
        \big]\T
  \label{eq:NARXstate}
\end{align}
that consists of the current and past outputs and inputs, and where $\My,
\Mu$ determine the NARX model order $\Nx = \My + \Mu + 1$.
The output is corrupted by Gaussian noise $\epsilon
\sim \CAL N(0,\Sn)$ with zero mean, noise variance $\Sn$, and bounded
support $\abs{\epsilon} \leq \bar\epsilon < \infty$.\footnote{In real
  processes the measurement noise is always bounded, for instance, due
  to the limitations of the involved data acquisition systems.
}
Inputs and outputs are restricted to lie in the constraint compact
sets $\U \subseteq \R$ and $\Y \subseteq \R$, where \U are hard
constraints and \Y can be hard or soft constraints that we denote by \Yh
and \Ys respectively.
The NARX state and the output are connected via $y_k =
\VEC c\T \xk$ with $\VEC c\T = \Matrix{1 &0 &\cdots &0}$.

The considered control objective is set-point stabilization and
optimal set-point change, \ie we want to steer the system from an initial
point $(\x_0,u_0)$ to a target reference
point $(\xref,\uref)$, while satisfying the constraints and
stabilizing the system at the target.
To this end, we employ \MPC, which requires a model 
\begin{align}
  \hykp = \fN(\xk,\uk)
\end{align} 
of the process \eqref{eq:NARX} that is capable of predicting
future output values with sufficient accuracy.
The hat notation $\hat{(\cdot)}$ denotes an estimated quantity.
We outline an approach to learn the system model approximation
$\fN(\xk,\uk)$ from measured input-output data using a \GP, which is
capable of online learning during operation based on newly available
data.
This results in a GP-based NARX prediction model.

\begin{remark}
  We consider a NARX model with one output that is modeled by a \GP.
  The presented approach can be extended to more outputs, where for each output
  an individual GP is used, c.f. Ostafew et
  al\cite{Ostafew2014,Ostafew2016} or Klenske et al\cite{Klenske2016}.
  The theoretical results obtained in
  Section~\ref{sec:model_predictive_control} are also valid for the
  multi-output case.
\end{remark}

\section{Gaussian Processes}
\label{sec:gaussian_processes}

We first review the basics of \GP regression and then present a
recursive formulation that is based on the concept of evolving GPs.
This facilitates the generation of a NARX prediction
model capable of adapting to changing conditions.
To reduce the online computational cost, we do not consider
online hyperparameter optimization.
Instead, we focus on updating the training data set efficiently and how
to perform the required computations online.
To this end, we combine this concept with a recursive update of the
involved Cholesky decomposition.

\subsection{Basics}

A \GP is a \emph{collection of random variables, any finite number of
which have a joint Gaussian distribution}. \cite{Rasmussen2006}
It generalizes the Gaussian probability
distribution to distributions over functions and can therefore be used
to model/approximate functions that can be used to capture dynamic
systems. \cite{Kocijan2016}
They can be utilized for models purely derived from data or 
combined in a hybrid way with other, for instance, deterministic models.
\cite{Ostafew2016, Hewing2018, Yang2015, McKinnon2017, Berkenkamp2017,
Akametalu2014, Berkenkamp2016a, Koller2018, Soloperto2018}

For regression, GPs are employed to derive or approximate maps of the
form $z = f(\VEC\nu) + \epsilon$ with input \VEC\nu, output $z$, and
where $f(\cdot)$ is the underlying but unknown \emph{latent} function.
The output is assumed to be corrupted by Gaussian noise\footnote{
  The concept of \GPs assumes Gaussian noise in the measurements, \ie
  noise with unbounded support.
  The considered real system \eqref{eq:NARXsystem}, however, is
  corrupted by Gaussian noise with bounded support.
  The resulting approximation error can be absorbed in the prediction
  error \eqref{eq:ep} defined further below.
  On the other hand, Gaussian noise with unbounded support can be
  regained by \emph{GP warping}\cite{Snelson2004}. 
  The smaller the bounded support, the larger the difference between the
  distributions and the larger the correcting effect of warping.
}
$\epsilon \sim \CAL N(0,\Sn)$ with zero mean and noise variance $\Sn$.
The objective is to infer the function $f(\cdot)$ using measured
input-output data $(\VEC\nu,z)$ with a \GP $g(\w)$ with input $\w \in
\R^\Nw$, called \emph{regressor}.
In the present case \eqref{eq:NARX}, we have $z = \ykp$ and
$f(\VEC\nu) = f(\xk,\uk)$.
The regressor of the GP will be $\w_k = (\xk, \uk) \in \R^\Nw$ with
regressor order $\Nw = \Nx + 1$.
For the sake of brevity we omit the dependence on the discrete time step
$k$ in the remainder of this section whenever possible.

The first required element is a GP prior distribution $g(\w) \sim
\CAL{GP} \big( m(\w), \cov(\w,\w') \big)$ that is specified via the mean
function $m(\w) = \E [g(\w)]$ and the covariance function\footnote{The
covariance function is also denoted as \emph{kernel}.
}
$\cov(\w,\w') = \text{cov}[g(\w),g(\w')] = \E
\big[\big(g(\w)-m(\w)\big) \big(g(\w')-m(\w')\big)\big]$
with $\w,\w' \in \R^\Nw$ and $\E[\cdot]$ denoting the expected
value.
The mean and covariance function together with a set of so-called
hyperparameters \thetaa, detailed later, fully specify the GP.

The GP prior is trained/learned using a set of $n$ measured
input-output data points, where the input data set is $\wb = [ \w_1
\ \cdots \ \w_n ]\T \in \R^{n \times \Nw}$ and the output data set 
$\zb = [ z_1 \ \cdots \ z_n ]\T \in \R^{n \times 1}$.
The combined data $\D = \{\wb,\zb\}$ is denoted as training data set
and is used to infer the posterior distribution
\begin{align*}
  g(\w|\D) \sim \CAL{GP} \big( \postM(\w|\D), \postV(\w|\D) \big) \ .
\end{align*}
This is also a \GP with posterior mean $\postM(\w|\D)$ and posterior
variance $\postV(\w|\D)$ given by
\begin{subequations}
\begin{align}
  & \postM(\w|\D) = m(\w) + \cov(\w,\wb) \KKi (\zb - m(\wb))
    \label{eq:postM} \\
  & \postV(\w|\D) = \cov(\w,\w) - \cov(\w,\wb) \KKi \cov(\wb,\w) \ ,
    \label{eq:postV}
\end{align}
\end{subequations}
with $m(\wb) = [ m(\w_1) \ \cdots \ m(\w_n) ]\T \in \R^{n \times 1}$,
$\cov(\w,\wb) = [ \cov(\w,\w_1) \ \cdots \ \cov(\w,\w_n) ] \in \R^{1 \times
n}$, $\cov(\wb,\w) = \cov(\w,\wb)\T$, and $\KK = \cov(\wb,\wb) =
[\cov(\w_i,\w_j)] \in \R^{n \times n}$.

Note that realizations of the posterior can yield infinitely many
function outcomes but as it is conditioned on the training data points,
it rejects all possible functions that do not go through or nearby (if
$\Sn \neq 0$) these points (Fig.~\ref{fig:GPexample}).

\begin{figure}[tb]
  \centering
  \includegraphics[width=0.6\linewidth]{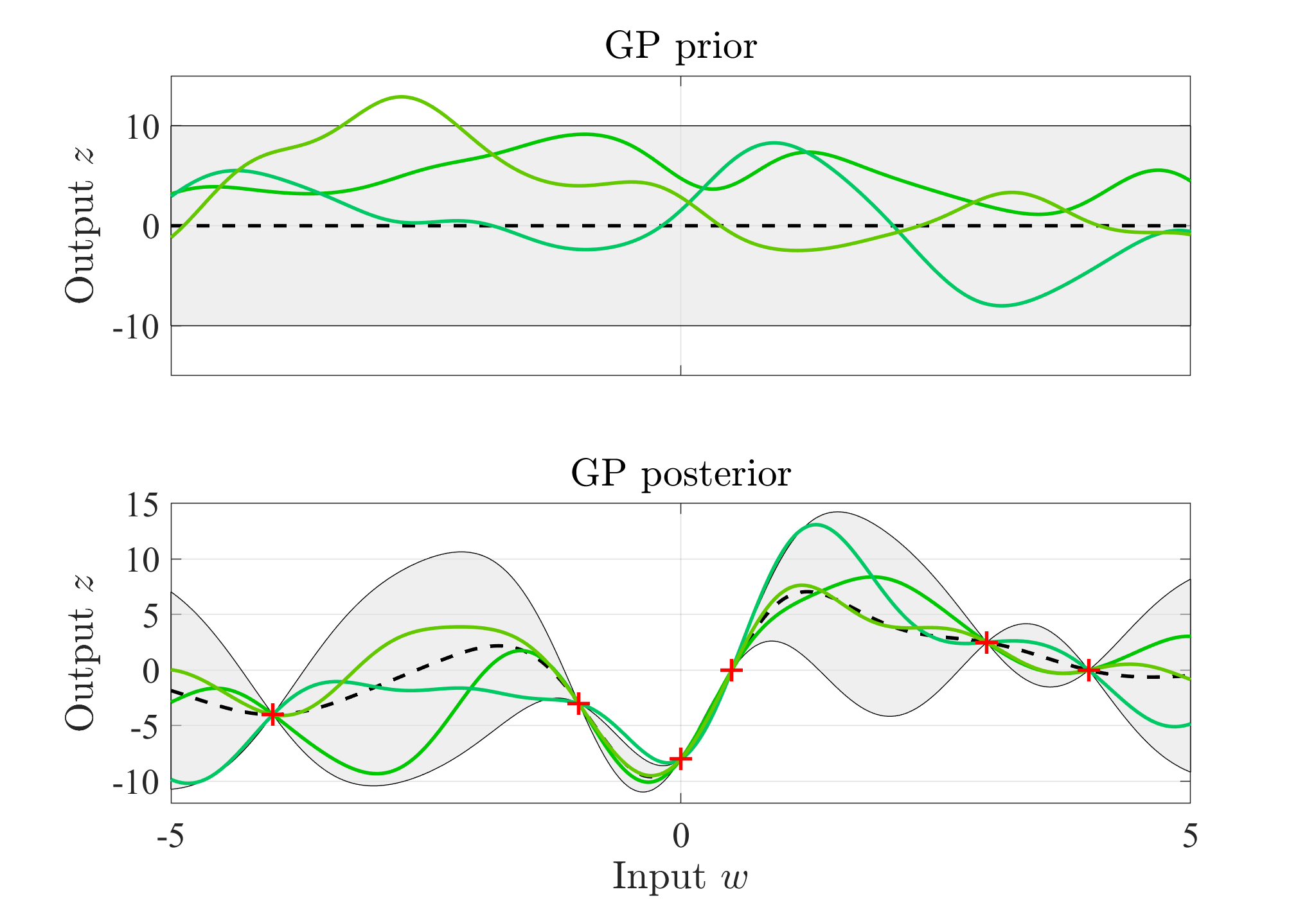}
  \caption{\GP inference: The top figure depicts a GP prior
  distribution with the dashed black line representing the mean
  function $m(\w)$ and the green lines representing random function
  realizations drawn from the prior distribution.
  The grey shaded area is the \unit[95]{\%} (twice the standard
  deviation) confidence interval computed via $\cov(\w,\w')$.
  When data points \D are added (bottom figure, red crosses), the GP
  posterior with $\postM(\w|\D)$ and $\postV(\w|\D)$ is inferred from
  this data.}
  \label{fig:GPexample}
\end{figure}

The posterior mean function \eqref{eq:postM} is the desired
estimator of the unknown output latent function $f(\xk,\uk)$ in
\eqref{eq:NARX}, which we highlight by defining
\begin{align}
  \hykp = \hat z := \fN(\xk,\uk) = \postM(\wk|\Dk) \ .
  \label{eq:GPpredictor}
\end{align}

The key elements for a \GP to yield a sensible model are the prior
mean and covariance function.
Both depend generally on a set of hyperparameters \thetaa, \ie
$m(\w|\thetaa)$ and $\cov(\w,\w'|\thetaa)$.
Very often just a constant zero prior mean $m(\w|\thetaa) = c =
0$ is used. \cite{Kocijan2003, Maiworm2018a, Gregorvcivc2009} 
However, other choices include, for instance, the use of a deterministic
base model $\xkp = f(\xk,\uk)$ as the prior mean function.
\cite{Yang2015, Roberts2013}
Regarding the covariance function, it is often assumed or known
that the system dynamics can be modeled by a member of the space of
smooth functions $C^\infty$.
A covariance function that provides this property is the \emph{squared
exponential covariance function with automatic relevance determination}
\begin{align}
  \cov(\w_i,\w_j|\thetaa) = \Sf \exp \left( -\frac{1}{2} (\w_i-\w_j)\T
  \Lambda (\w_i-\w_j) \right) + \Sn \delta_{ij} ,
\label{eq:covSEard}
\end{align}
where $\w_i,\w_j \in \R^{\Nw}$, $\thetaa = \{ \Sf, \Lambda \}$, and
$\Lambda = \text{diag}(l_1^{-2}, \ldots, l_{\Nw}^{-2})$.
The measurement noise \Sn is added via the Kronecker delta $\delta_{ij}$
in \eqref{eq:covSEard}.
The minimal required number of regressors \Nw can be determined through
optimization of the length scale parameters $l$ in $\Lambda$.
\cite{Williams1996, Kocijan2005a}
Other choices include, for instance, the combination of
\eqref{eq:covSEard} with a linear kernel.
\cite{Williams1996, Ackermann2011}

A common approach to determine the hyperparameters $\thetaa$, given a
training data set $\D = \{\wb,\zb\}$, is to maximize the log marginal
likelihood \cite{Rasmussen2006}
\begin{align}
  \log\big(p(\zb|\wb,\thetaa)\big) = &-\frac{1}{2}\zb\T \KKi \zb
    -\frac{1}{2}\log |\KK| -\frac{n}{2} \log(2\pi) \ . 
    \label{eq:loglike}
\end{align}

An advantage of \GPs is that \eqref{eq:postV} naturally provides a
quantification of the model uncertainty in the form of its variance.
On the other hand, the involved computations in \eqref{eq:postM} and
\eqref{eq:postV} scale with $\Op(n^3)$ due to \KKi, where $n$ is the
number of training data points.
This severely limits the application of GP models for fast processes,
where small sampling times are required;
especially in the case of relatively large training data sets with several
hundred or thousands of data points.
If online or close to online hyperparameter optimization is needed, this
drawback becomes even more pronounced.

\subsection{Evolving Gaussian Processes}
\label{sec:evolving_GPs}

In order to efficiently refine the GP model online we seek to update the
training data set \Dk, possibly at each time step $k$, during operation.
To this end, we resort to the concept of so-called \emph{evolving}
GPs\cite{Petelin2011, Kocijan2016}, which can be used, for instance,
if the training data is only available for certain regions
of the operating space and one wants to expand operation beyond these
regions online.
The concept basically leads to GPs whose training data set \Dk is
updated online using some type of information criterion.
Different criteria can be used to select new data points to be added
and already existing points to be removed if necessary.

The general idea is to include an incoming data point to the training
data set only if it contributes enough new valuable information, which
can be defined in different ways and depends on the respective
application.
Possible options are the use of the information gain, entropy
difference, or the expected likelihood. \cite{Smola2001, Seeger2003}
We employ the GP as a prediction model in MPC and are
therefore particularly interested in how accurate the current
model is able to predict the output value at the next time step and how
confident this prediction is.
To this end, given a new data point $(\wk,\ykp)$, we first define the
prediction error via
\begin{align}
  \ep := \ykp - \hykp = f(\xk, \uk) + \epsilon - \postM(\wk|\Dk) \ ,
  \label{eq:ep}
\end{align}
and define the following rule that determines a new training data set
candidate \Dskp.

\begin{LWrule}[New training data set candidate]
  At the current time step $k$ with regressor \wk and training data set
  \Dk compute $\hykp = \postM(\wk|\Dk)$ and $\postV = \postV(\wk|\Dk)$.
  Once the next output \ykp is available, the new data point
  $(\wk,\ykp)$ is considered as a candidate for inclusion into the
  training data set \Dk
  \begin{algorithmic}
    \If{ $\abs{\ep} > \epU$ \ OR \ $\postV > \sigmaU$ }
    \State $\Dskp = \Dk \cup (\wk,\ykp)$
    \EndIf
  \end{algorithmic}
  where \epU and \sigmaU are pre-specified thresholds and \Dskp is the new
  training data set candidate for $k+1$.
  \label{rule:add_data_point}
\end{LWrule}

Thus, if the prediction error \ep is larger then the threshold \epU,
the data point is considered to be included in the training data set \Dk
because the current posterior model is not able to predict the output
with the specified accuracy.
If it is smaller but the resulting posterior variance $\postV(\w|\Dk)$ is
larger than the threshold $\sigmaU$, the data point is also a candidate
because the current posterior model is not sufficiently confident in its
prediction.
This allows to include data points that are relevant to
attain a certain prediction quality and effectively allows to limit the
necessary number of data points in \Dk.
This becomes especially important for long operation times and many
encountered data points with new information during operation.

\begin{remark}
  Since Rule~\ref{rule:add_data_point} would also consider
  outliers for inclusion, we propose to combine it with an additional
  update rule presented in \cref{thm:nominal_stability}
  (\cref{sec:stability}).
  The application of both update rules is contained in \cref{alg1}.
\end{remark}

As the available computational power is always limited and depending on
the concrete system, this can require the limitation of the maximum
number of points in \Dk by a constant $M \in \N$.\footnote{This approach
is also sometimes denoted as \emph{truncated} GP\cite{Carron2016}.}
If this limit is reached, data points have to be removed to maintain the
size of \Dk.
Again, different criteria can be employed to determine which data point
shall be deleted.
For instance, the point in the training data set with the lowest
benefit for the model quality (\eg the data point that is most
accurately predicted under the current posterior) can be deleted.
This, however, can be computationally expensive because the prediction
has to be evaluated for every of the $M$ training data points at each
time instant $k$.
For online implementation, we employ a more simple approach that
deletes the oldest point contained in \Dk.

\begin{remark}
  The concept of evolving GPs, in particular the outlined data handling
  approach, leads to a training data set \Dk that captures the
  system dynamics in an (evolving) subregion of the whole operating
  region.
  Thus, information about already visited
  regions can be lost when moving towards other regions and have to be
  regained when visited again.
  This could be counteracted, for instance, by exploiting multiple GPs
  for different regions or by GP blending. \cite{Bethge2018}
\end{remark}

\begin{remark}
  In principle, the smaller the thresholds \epU and \sigmaU, the better
  the prediction.
  However, then also the overhead for the computational
  evaluation for adding and removing data points becomes larger.
  In addition, the smaller the thresholds, the smaller the region in
  which the training data set captures the system behavior, given the
  case that only a finite number of training data points is allowed.
  Hence, the selection of the thresholds \epU and \sigmaU is an
  application specific trade-off and might be chosen heuristically by
  the user.
  Some general guidelines are, $(i)$ 
  a lower bound for \sigmaU is the measurement noise variance, and $(ii)$
  \epU could be chosen proportional to $\frac{1}{M} \sum_{i=1}^M \abs{\zb -
  \postM(\wb|\Dk)}$, \ie to the mean value of all the absolute values of
  the prediction errors, based on the current training data set \Dk.
  In the same way \sigmaU could be chosen.
\end{remark}

\subsection{Avoiding Numerical Ill Conditioning for MPC by Cholesky
Decomposition}
\label{sec:cholesky_decomposition}

The squared exponential covariance function \eqref{eq:covSEard} and
other smooth covariance functions lead to a poor conditioned
covariance matrix \KK. \cite{Neal1997, Osborne2010}
This results in numerical problems when computing the inverse \KKi
with computational cost $\CAL O(n^3)$, as
required for \eqref{eq:postM}, \eqref{eq:postV}, or \eqref{eq:loglike}.
These problems become even worse if \eqref{eq:postM} and
\eqref{eq:postV} are nested within an optimization procedure like \MPC.
One way to alleviate this problem is by adding an additional noise or
\emph{jitter} term \cite{Neal1997} to the diagonal of
the covariance matrix.
An effective approach however is to avoid the numerical
instabilities that arise in the explicit computation of the
matrix inverse by performing the required computations using the
Cholesky decomposition, which is numerically more stable.

Given a system of linear equations $\A \x= \VEC b$ with a
symmetric positive matrix $\MAT A$, we denote the solution by $\x =
\A^{-1} \VEC b := \A \backslash \VEC b$.
The Cholesky decomposition of $\A$ is $\A = \RR\T \RR$, where
$\RR = \text{chol}(\A)$ is an upper triangular matrix that is called the
\emph{Cholesky factor}.
It can be used to obtain the solution via $\x = \RR \backslash (\RR\T
\backslash \VEC b)$.
In order to use the Cholesky factor to solve \eqref{eq:postM} and
\eqref{eq:postV}, we define
\begin{align*}
  \alphaa &:= \KKi (\zb - m(\wb)) \\
  \betaa  &:= \KKi \cov(\wb,\w) \ ,
\end{align*}
which can then be computed with the Cholesky decomposition $\KK = \RR\T
\RR$ via
\begin{align}
  \alphaa &= \RR \backslash \left(\RR\T \backslash (\zb - m(\wb) )
    \right) \label{eq:alphaa} \\
  \betaa  &= \RR \backslash \left(\RR\T \backslash \cov(\wb,\w) \right)
    \ . \nonumber
\end{align}

The computational cost of computing \RR is $\Op(\nicefrac{n^3}{6})$
and the cost of computing \alphaa and \betaa is $\Op(n^2)$.
\cite{Rasmussen2006}

If the training data set \Dk does not change, the
Cholesky decomposition $\KK = \RR\T \RR$ and the computation of \alphaa
have to be performed only once at the beginning, whereas
\betaa has to be recomputed for every new test point \w.
If \Dk changes, \ie with each inclusion or removal of a
data point, the covariance matrix \KK has to be updated for an
appropriate evaluation of the GP posterior.
If a data point is included, a row and column have to be added to \KK.
If a data point is removed, the respective row and column associated
with this point have to be removed.
These changes require in principle a full recalculation of the Cholesky
factor \RR, which is the most expensive computation.
To reduce this computational load we employ the approach of
Osborne\cite{Osborne2010} to recalculate the Cholesky factor
recursively, taking advantage of the available factor of the
previous step.
The precise procedure is outlined in the appendix in
Sec.~\ref{sec:recursive_cholesky_factor_update}.

\begin{remark}
  The recursive update of the Cholesky factor can only be applied if the
  hyperparameters \thetaa do not change because otherwise, every single
  element of \KK changes and a recursive approach is not applicable
  anymore.
\end{remark}

\begin{remark}
  Note that in many works\cite{Huber2014, Vaerenbergh2012, Perez2013} not the
  Cholesky decomposition but the covariance matrix inverse \KKi is
  recursively computed, which is based on the partitioned block
  inverse using the \emph{Woodbury matrix identity}.
  Presumably for the numerical issues outlined above, this approach has
  never been used in combination with MPC.
  It has, however, in the signal processing literature,
  where it is strongly connected to the concept of \emph{kernel recursive
  least-squares}. \cite{Vaerenbergh2012, Perez2013}
\end{remark}

Due to the recursive nature, both in the data inclusion approach and the
Cholesky decomposition, we denote the resulting \GP as recursive GP
(rGP).
The most important steps of the resulting rGP-MPC formulation are
presented in Algorithm~\ref{alg1}.

\section{Gaussian Process Based Output Feedback Model Predictive Control}
\label{sec:model_predictive_control}

In this section, we present the output feedback \MPC formulation, based
on the rGP NARX model for prediction.
We highlight the necessary components and show under which
conditions stability can be guaranteed even if the GP model changes
online.

\subsection{Prediction Model}
\label{sec:prediction_model}

In \cref{sec:stability} we establish input-to-state stability
for the considered system, which is defined using the evolution of
the state and not the output.
For this reason, we first reformulate the GP output prediction in terms
of the NARX state \hxk.
We start by setting $k := k+1$ in \hxk and arrive at
\begin{align*}
  \hxkp = \big[\hykp, y_k, \ldots, y_{k+1-\My}, \uk,
            \ldots, u_{k+1-\Mu}\big] \ .
\end{align*}
Since the predicted output \hykp is computed by \eqref{eq:GPpredictor} we
obtain the NARX prediction model
\begin{align}
  \hxkp = \FN(\hxk, \uk | \Dk)
        := \big[\postM(\wk|\Dk), y_k, \ldots, y_{k+1-\My}, \uk,
            \ldots, u_{k+1-\Mu}\big] \ ,
\label{eq:predictionModel}
\end{align}
which we also denote as the \emph{nominal} model.
~\\
Correspondingly, for the NARX model of the real process \eqref{eq:NARX}
we have
\begin{align*}
  \xkp &= \big[y_{k+1}, \yk, \ldots, y_{k+1-\My}, \uk, \ldots,
    u_{k+1-\Mu}\big] \\
       &= \big[f(\xk,\uk) + \epsilon, \yk, \ldots, y_{k+1-\My}, \uk,
       \ldots, u_{k+1-\Mu}\big]
\end{align*}
and due to \eqref{eq:ep} this can be reformulated as
\begin{align}
\begin{aligned}
  \xkp &= \big[\postM(\wk|\Dk) + \ep, \yk, \ldots, y_{k+1-\My}, \uk,
        \ldots, u_{k+1-\Mu}\big] \\
       &= \FN(\xk, \uk | \Dk) + \VEC d \ep =: F(\xk,\uk,\ep) 
\end{aligned}
  \label{eq:NARXstateModel} 
\end{align}
with $\VEC d = \Matrix{1 &0 &\cdots &0}\T$, \ie the real NARX model can
be represented as the superposition of the nominal/prediction model and
the prediction error.

\subsection{MPC Optimization Problem}
\label{sec:optimal_control_problem}

Using the prediction model \eqref{eq:predictionModel}, we consider at
each time step $k$ the optimization problem
\begin{align}
\begin{aligned}
  \min_{\uS} \ & \VN\big(\xk,\uS\big) \\
  \text{s.t.~} &\forall i \in \mathcal I_{0:N-1}\!: \\
  & \hx_{k+i+1|k} = \FN \left( \hx_{k+i|k}, \hu_{k+i|k} |\Dk \right) \\
  & \hx_{k|k} = \xk \\
  & \hu_{k+i|k} \in \U \\
  & \hx_{k+i|k} \in \X \ .
\end{aligned} 
  \label{eq:OCP}
\end{align}
The input sequence to be optimized is denoted by $\uS = \big\{
  \hu_{k|k}, \ldots, \hu_{k+N-1|k} \big\}$, $N$ is the prediction
horizon, \xk is the initial condition of the measured NARX state
\eqref{eq:NARXstate}, and $\X \subseteq \R^\Nx$ is the resulting constrained
set of the NARX state that is a combination of multiple instances of \Yh
depending on the specific composition of \xk.\footnote{If for instance
  $\xk = [\yk, y_{k-1}, y_{k-2}]$, then $\X = \Yh \times \Yh \times \Yh$.
  }
Since \Yh is compact, the resulting \X is also compact.
As cost function in \eqref{eq:OCP} we consider
\begin{align*}
  \VN\big(\xk,\uS\big) = \sum_{i=0}^{N-1} \ell \big( \hx_{k+i|k},
  \hu_{k+i|k}\big) + \lambda \Vf \big( \hx_{k+N|k} - \xref \big) \ ,
\end{align*}
where $\Vf(\cdot)$ is the terminal cost function that is weighted by a
design parameter $\lambda \geq 1$.
The employed positive stage cost is given by
\begin{align*}
  \ell(\hxk,\hu_k) = \ls(\hxk-\xref, \hu_k-\uref) + \lb(\hyk) \ ,
\end{align*}
where $\ls(\cdot)$ penalizes input and state deviations from the
reference and $\lb(\cdot)$ is a barrier function that can
account for soft output constraints \Ys.  
It is defined by
\begin{align*}
  \lb(\hyk) \geq \alpha_\text b \big( d(\hyk,\Ys) \big) \ ,
\end{align*}
and must satisfy $\lb(\hyk) = 0, \forall \hyk \in \Ys$, where
$\alpha_\text b(\cdot)$ is a \K-function and $d(\cdot)$ the distance
function as defined in Section~\ref{sec:introduction}.

The optimal solution of \eqref{eq:OCP} is denoted by $\uS^*$, the
resulting optimal state sequence by $\xF^*_{k|k}$.
The first element of $\uS^*$, i.e.  $\hu^*_{k|k}$, is applied to the
process such that we obtain $\uk = \kMPC(\xk|\Dk) = \hu^*_{k|k}$.
Note that the implicitly defined control law $\kMPC(\xk|\Dk)$ is
time-varying, as well as the resulting optimal cost function
$\VN^*(\xk|\Dk) = V_N(\xk,\uS^*|\Dk)$, also denoted as \emph{value
function}, because they depend on the changing prediction model
associated with \Dk.
Note furthermore that \eqref{eq:OCP} does not include any explicit
terminal region constraint for stability.
This makes its solution less computationally expensive, especially if only
soft output/state constraints are considered.

\subsection{Stability}
\label{sec:stability}

Establishing stability in MPC is often based on the use of a
terminal cost function $\Vf(\cdot)$ and a terminal region
\Xf.\cite{Mayne2000}
Here we employ an approach where the \OCP \eqref{eq:OCP} does
not require an explicit terminal region \Xf.
Instead, we use $\Vf(\cdot)$ weighted by a factor $\lambda$, as proposed
by Limon et al\cite{Limon2006}, to establish input-to-state stability.

\begin{definition}[Input-to-state Stability]
  Consider the closed-loop system $\xkp =
  F\big( \xk,\kMPC(\xk|\Dk),\ekV \big)$.
  The set-point \xref is input-to-state stable (ISS) if there exist a
  \KL-function $\beta(\cdot,\cdot)$ and a \K-function $\gamma(\cdot)$
  such that 
  \begin{align}
    \norm{\xk - \xref} \leq \beta(\norm{\x_0-\xref},k) 
      + \gamma\Big(\max_{k \geq 0}\norm{\ekV}\Big)
  \end{align}
holds for all initial states $\x_0$, errors \ekV, and for all $k$.
\end{definition}

ISS combines nominal stability as well as
uniformly bounded influence of uncertainty in a single condition.
It implies asymptotic stability of the undisturbed (nominal) system (with
$\ekV \equiv 0$) and a bounded effect of the uncertainty on the state
evolution.
Furthermore, if the error signal \ekV fades, the uncertain system
asymptotically converges to the reference point.
We therefore consider stability first for the nominal case, \ie when the
prediction/nominal model \eqref{eq:predictionModel} and the true
system \eqref{eq:NARXstateModel} are exactly the same.
After that, we establish robust stability in the sense of input-to-state
stability.

\subsubsection{Nominal Stability}

In the following, let the current deviation from the reference point and
the deviation at the next time step be $\xE = \x - \xref$ and $\xEp =
\x^+ - \xref$ respectively.
This change of coordinates is required if $\xref \neq 0$.

\begin{assumption}
  Assume that
  \begin{enumerate}
    \item the stage cost function $\ell(\x,u)$ is positive definite,
      \ie $\ell(\xref,\uref)=0$ and there exists a \Kinfty-function
      $\alpha(\cdot)$ such that $\ell(\x,u) \geq \alpha(\norm{\xE})$
      for all $u \in \U$, and 
    \item there exists a terminal control law $\kf(\cdot)$ and 
      a control Lyapunov function $\Vf(\cdot)$ such that the conditions
      \begin{gather*}
        \alpha_1(\norm{\xE}) \leq \Vf(\xE) \leq \alpha_2(\norm{\xE}) \\
        \intertext{and}
        \Vf\big(\xEp\big) \leq \Vf(\xE)
          - \ell\big(\xE+\xref,\kf(\xE)+\uref \big)
      \end{gather*}
      hold for all $\xE \in \Xf = \{ \xE \in \R^\Nx: \Vf(\xE) \leq \nu
      \} \subseteq \X$ with $\nu > 0$ and $\xEp = \FN\big(\xE+\xref,
      \kf(\xE)+\uref | \D\big) - \xref$, and where $\alpha_1(\cdot)$ and
      $\alpha_2(\cdot)$ are \Kinfty-functions.
      The constant $\nu$ is chosen such that $\Xf \subseteq \X$ and
      $\kf(\xE) + \uref \in \U$ for all $\xE \in \Xf$.
  \end{enumerate}
  \label{ass:cost}
\end{assumption}

Assumption~\ref{ass:cost} ensures that the system is locally
asymptotically stable on the positive invariant set \Xf, while
satisfying state and input constraints.
It can be satisfied if we have at least a locally
valid description of the process at the target point.
This can, for instance, be a linearized version of the GP prediction
model at the reference (\eg $\D = \Dref$), which in turn can then be
used to derive a suitable terminal cost and controller.
Possible options are then, for instance, the use of a linear-quadratic
regulator and/or applying Lyapunov methods (see also
Sec.~\ref{sec:terminal_ingredients}).
Although it is sufficient to determine the terminal components from the
nominal model, one could also consider the design of a robust terminal
controller and cost.
For instance, using a GP model for the target region one could consider a
specific probability bound given by the posterior variance and then
based on this design a robust terminal controller.

\begin{theorem}[Nominal stability]
  \label{thm:nominal_stability}
  Let $\kMPC(\xk|\Dk)$ be the predictive controller derived from the
  optimal control problem \eqref{eq:OCP} and let
  Assumption~\ref{ass:cost} be satisfied.
  Furthermore, let \Dk be the training data set at time $k$, \Dkp the
  one that will be used at time $k+1$, and \Dskp the updated training
  data set candidate resulting from Rule~\ref{rule:add_data_point}.
  If \Dk is updated using the additional rule
  \begin{algorithmic}
    \If{ $V_N^*\big( \xk | \Dskp \big) \leq V_N^*\big( \xk | \Dk \big)$ }
    \State $\Dkp \leftarrow \Dskp$
    \Else 
    \State $\Dkp \leftarrow \Dk$ 
    \EndIf
  \end{algorithmic}
  then $\forall \lambda \geq 1$, there exists a feasible region $\XNo
  \subseteq \X$ such that $\forall \x_0 \in \XNo$ the target \xref
  of the nominal closed-loop system $\xkp = \FN\big(\xk,
  \kMPC(\xk|\Dk)\big)$ is asymptotically stable.
  The size of the set $\XNo$ increases with $\lambda$.
\end{theorem}

\begin{proof}[Proof]
  Let $\xF^*_{k|k} = \big\{ \hx^*_{k|k}, \hx^*_{k+1|k}, \ldots,
  \hx^*_{k+N|k} \big\}$ be the predicted state sequence that results
  from applying the optimal input sequence $\uS^*$.
  Then we can write the optimal cost for initial condition $\xk =
  \hx^*_{k|k}$ also as $V_N^*\big(\xk|\Dk\big) =
  V_N\big(\xk,\uF^*_{k|k}|\Dk\big) =
  V_N\big(\xF^*_{k|k},\uF^*_{k|k}|\Dk\big)$.
  Let furthermore 
  $\xF^*_{k+1|k} = \big\{ \hx^*_{k+1|k}, \ldots, \hx^*_{k+N+1|k} \big\}$ 
  and 
  $\uSp^* = \big\{ \hu^*_{k+1|k}, \ldots, \hu^*_{k+N-1|k}, \kf\big(
    \hx^*_{k+N|k} - \xref \big) + \uref \big\}$
  be the respective sequences that start at $k+1$ computed at time $k$,
  where the last state is given by the terminal control law, \ie
  $\hx^*_{k+N+1|k} = \FN\big(\hx^*_{k+N|k}, \kf\big(\hx^*_{k+N|k} -
  \xref \big) + \uref | \Dk \big)$.

  By Assumption~\ref{ass:cost} we have that the stage and terminal
  cost are positive definite.
  Hence, the cost function $V_N\big(\xk,\uS\big)$ is positive definite.
  Furthermore we also obtain
  \begin{align*}
    V_N\big(\xF^*_{k+1|k},\uF^*_{k+1|k}|\Dkp\big) \leq
      V_N\big(\xF^*_{k|k},\uF^*_{k|k}|\Dkp\big) - \ell(\xk,\uk)
  \end{align*}
  by Assumption~\ref{ass:cost}, which is a well known result in standard
  MPC (for the derivation see, for instance, Rawlings and
  Mayne\cite{Rawlings2009} or Rakovic et al\cite{Rakovic2019}).
  Given the update rule in Theorem~\ref{thm:nominal_stability} we have
  \begin{align*}
    V_N\big(\xF^*_{k|k},\uF^*_{k|k}|\Dkp\big) - \ell(\xk,\uk)
     \leq V_N\big(\xF^*_{k|k},\uF^*_{k|k}|\Dk\big) -
      \ell(\xk,\uk) \ .
  \end{align*}
  Combining the previous two equations we obtain
  \begin{align}
    V_N\big(\xF^*_{k+1|k},\uF^*_{k+1|k}|\Dkp\big) \leq
      V_N\big(\xF^*_{k|k},\uF^*_{k|k}|\Dk\big) - \ell(\xk,\uk) \ .
      \label{eq:VN_decrease}
  \end{align}
  Thus, the value function is decreasing even if the prediction model
  changes.
  Hence the value function is a Lyapunov function.

  Regarding the feasible region, we first review a result of Limon et
  al\cite{Limon2006} and show afterwards an extension to the
  present case.
  In particular, Theorem~3 of Limon et al\cite{Limon2006} shows for the 
  nominal and time-invariant case of \eqref{eq:OCP} (\ie constant
  prediction model and no model-plant mismatch) with value function
  $\VN^*(\xk)$ that $\forall \lambda \geq 1$ there exists a feasible
  region \XN such that $\forall \x_0 \in \XN$ the nominal
  closed-loop system $\xkp = \FN\big(\xk, \kMPC(\xk)\big)$ is
  recursively feasible and asymptotically stable.
  The feasible region is characterized by $\XN
  = \left\{ \xk \in \R^\Nx: V^*_N(\xk) \leq N \cdot d + \lambda \cdot
  \nu \right\}$, where $\nu$ is defined in Assumption~\ref{ass:cost} and
  $d$ is a positive constant such that $\ell(\xk,\uk) > d$, $\forall \xk
  \notin \Xf$ and $\forall \uk \in \U$.
  The size of the set \XN increases with $\lambda$.\footnote{
    Note that Theorem~3 of Limon et al\cite{Limon2006} is stated the
    other way round, \ie for each region \XN and for all $\xk \in
    \XN$, there exists a $\lambda \geq 1$ such that the nominal
    closed-loop system is asymptotically stable at \xref.
    }
  
  In this work, the value function $\VN^*(\xk|\Dk)$ changes at
  certain time instances $k$ whenever the data set \Dk changes.
  For this reason we extend the definition of the feasible region to
  $\XNk = \left\{ \xk \in \R^\Nx: V^*_N(\xk|\Dk) \leq N \cdot d + \lambda
  \cdot \nu \right\}$, which then also changes with $k$.
  Due to \eqref{eq:VN_decrease}, the optimal cost is
  decreasing for a particular state sequence $\mathbf{x} = \{ \x_0,
  \x_1, \ldots, \xk, \ldots \}$ with $N d + \lambda \nu \geq
  V_N^*(\x_0|\D_0) \geq V_N^*(\x_1|\D_1) \geq \ldots \geq
  V_N^*(\xk|\Dk)$ and therefore \XNk is increasing along the state
  sequence.
  Thus, if the initial state $\x_0 \in \XNo$, then the subsequent states
  $\xk \in \XNk$ and the \OCP is recursively feasible.
  Hence, the target \xref is asymptotically stable for the nominal
  closed-loop system $\xkp = \FN\big(\xk, \kMPC(\xk|\Dk)\big)$.
\end{proof}

At \xk (with the current output measurement \yk) the optimal
control problem is solved with data set \Dk
and the resulting input $\uk = \kMPC(\xk |\Dk) = \hu^*_{k|k}$ is applied
to the system.
If the next data point $(\wk,\ykp)$ is a candidate for updating the GP, the
previous optimal cost is recomputed using the updated GP.
If the cost does not increase, the GP update becomes effective.
Thus, the update rule in Theorem~\ref{thm:nominal_stability} is
executed additionally after the data selection process of
Rule~\ref{rule:add_data_point}.
This is also reflected in Algorithm~\ref{alg1}.

\begin{remark}[(Conflicting objectives)]
  \label{rem:properties_update_rule}
  \cref{thm:nominal_stability} establishes nominal
  stability despite a changing training data set \Dk.
  In order to determine the new data set candidate \Dskp we use
  Rule~\ref{rule:add_data_point}, whose objective is to refine the
  current prediction model.
  Note that also other rules, which utilize different selection criteria
  for model refinement (e.g. statistical methods, see
  \cref{sec:evolving_GPs}), can be employed.
  Now, one could assume that the additional update rule in
  \cref{thm:nominal_stability} is not necessary because with every new
  data point the prediction model should become more accurate.
  This is, however, not necessarily the case if, for instance, the
  output is corrupted by noise or if outliers are present.
  In both cases, the apparent process behavior differs from the true
  behavior and it cannot be guaranteed that the prediction model becomes
  more accurate with every added data point, nor that the value function
  continues decreasing monotonically.
  Thus, the objective of the update rule in \cref{thm:nominal_stability}
  is to make sure that safety, in the sense of stability and constraint
  satisfaction, is guaranteed.
  This is also illustrated in the simulations section in
  \cref{fig:sim_7_updateRule}.
  However, in the same simulations we also see that data points,
  selected by Rule~\ref{rule:add_data_point} and which carry valuable
  information, are discarded by the update rule of
  \cref{thm:nominal_stability}
  because the decreasing value function condition, and with that
  stability, could not be guaranteed.
  In other words, the two objectives of model refinement (expressed by
  Rule~\ref{rule:add_data_point}) and safety (in the sense of stability,
  expressed by the update rule of \cref{thm:nominal_stability}) are
  conflicting objectives, especially in the case of corrupted
  measurements.
  In this work we prioritize safety, thereby sacrificing a bit of the
  potential of model refinement.
\end{remark}

On the basis of the nominal stability result for the online rGP-MPC
scheme, we now establish robust stability.

\subsubsection{Robust Stability}
\label{sec:robust_stability}

Based on Theorem~\ref{thm:nominal_stability} we show that the real
process controlled by the proposed predictive controller is
input-to-state stable w.r.t. the prediction error \ep.

\begin{theorem}[Input-to-state Stability]
  \label{thm:robust_stability}
  Let $\kMPC(\xk|\Dk)$ be the predictive controller derived from optimal
  control problem \eqref{eq:OCP} satisfying Assumption~\ref{ass:cost}
  and Theorem~\ref{thm:nominal_stability}.
  If
  \begin{itemize}
    \item the nominal model $\FN\big(\xk,\uk|\Dk \big)$ is \uc in \xk
      for all $\xk \in \XNo$, all $\uk \in \U$, and all \Dk during the
      prediction horizon\footnote{Note that this condition does not
        prohibit the change of the nominal model from the current time
        instant $k$ to the next $k+1$.
        }, and
    \item the stage cost function $\ell(\xk,\uk)$ and the terminal cost
      function $\Vf(\xk)$ are \uc in \xk for all $\xk
      \in \XNo$ and all $\uk \in \U$,
  \end{itemize}
  then the target \xref of the closed-loop system $\xkp =
  F\big(\xk,\kMPC(\xk|\Dk),\ep \big)$ is ISS w.r.t. the prediction error
  {\ep} in a robust feasible set $\OOo \subset \XNo$ for a sufficiently
  small $\mu$ with $\abs{\ep} < \mu < \infty$. 
  The smaller $\mu$, the larger the set \OOo.
\end{theorem}

\begin{proof}
  We first establish the set \OOo and prove recursive feasibility.
  Afterwards we prove the ISS property.
  
  Regarding the nature of \OOo we first review a result of Limon et
  al\cite{Limon2009} and then extend it to our case.
  Proposition~1 (C2) in\cite{Limon2009} shows for the time-invariant
  case of \eqref{eq:OCP} (\ie for a non-changing prediction model)
  that the closed-loop $\xkp = F(\xk,\kMPC(\xk|\Dk),\ep)$ is robustly
  feasible for all \xk in a robust feasible set \OO.
  In particular, it is proven that if $\abs{\ep} < \mu$ with a
  sufficiently small $\mu$, there exists a $r$ such that $\OO := \{ \xk
  \in \R^\Nx : V_N^*(\xk) \leq r \} \subset \XN$ is a compact and
  positive invariant set (where \XN is the feasible set of the OCP with
  $\ep \equiv 0$) and for all $\xk \in \OO$ the resulting predicted
  state sequence remains in \OO.
  Therefore the state constraints \X do not become active.
  Hence, for all $\x_0 \in \OO$ the MPC scheme is recursively
  feasible and the constraints are robustly satisfied.
  Furthermore, larger values of $\lambda$ lead to a larger region \OO.
  
  In the definition of \OO in \cite{Limon2009} the value function
  $\VN^*(\xk)$ is time-invariant, whereas in this work $\VN^*(\xk|\Dk)$
  depends on the changing data set \Dk.
  For this reason we extend the definition of the robust feasible region
  to $\OOk = \left\{ \xk \in \R^\Nx: V^*_N(\xk|\Dk) \leq r \right\}
  \subset \XNk$, which then also changes with $k$.
  In order for $\OOk \subset \XNk$ to hold we require $0 < r < N
  \cdot d + \lambda \cdot \nu$ because $\XNk = \left\{ \xk \in \R^\Nx:
  V^*_N(\xk|\Dk) \leq N \cdot d + \lambda \cdot \nu \right\}$.
  Thereby, $N \cdot d + \lambda \cdot \nu$ establishes an upper bound
  for $r$.
  Like the feasible set \XNk (see the proof to
  \cref{thm:nominal_stability}), also \OOk increases with $\lambda$ and
  in particular with $k$ along a particular state sequence $\mathbf{x} =
  \{ \x_0, \x_1, \ldots \}$.
  Therefore, if the initial state $\x_0 \in \OOo$, then the subsequent
  states $\xk \in \OOk$ and the state constraints do not become active.
  The existence of \OOo is established by Proposition~1 (C2)
  in\cite{Limon2009} (as outlined above) and therefore, if $\x_0 \in
  \OOo$ then \eqref{eq:OCP} is recursively feasible and the constraints
  are robustly satisfied.
  
  Now we show that the closed-loop system $\xkp =
  F(\xk,\kMPC(\xk|\Dk),\ep)$ is ISS w.r.t. the prediction
  error \ep.
  To this end we start by showing that the cost function
  $\VN(\xk,\uS)$ is \uc in \xk.
  Since the nominal model $\FN(\xk,\uk|\Dk)$ is \uc in \xk
  during the prediction horizon, there exists a \K-function
  $\sigma_x(\cdot)$ such that $\norm{\FN(\xk,\uk|\Dk) -
  \FN(\zk,\uk|\Dk)} \leq \sigma_x(\norm{\xk-\zk})$ for all $\xk,\zk
  \in \XNo$, all $\uk \in \U$, and for a given data set \Dk.
  In accordance with Lemma~2 in \cite{Limon2009}, the predicted state
  evolution then satisfies $\norm{\hx_{k+i|k} - \hat\z_{k+i|k}} \leq
  \sigma_x^i(\norm{\xk - \zk})$ for $i \in \CAL I_{0:N-1}$.
  Furthermore, since the stage and terminal cost are \uc in \xk, there
  exists a couple of \K-functions $\sigma_\ell(\cdot),
  \sigma_{\Vf}(\cdot)$ such that $\norm{\ell(\xk,\uk) - \ell(\zk,\uk)}
  \leq \sigma_\ell(\norm{\xk-\zk})$ and $\norm{\Vf(\xk) - \Vf(\zk)} \leq
  \sigma_{\Vf}(\norm{\xk-\zk})$ for all $\xk,\zk \in \XNo$ and all $u
  \in \U$.
  Combining these properties we obtain
  \begin{align*}
    \norm{\VN(\xk,\uS) - \VN(\zk,\uS)}
    &\leq \sum_{i=0}^{N-1} \norm{\ell(\hx_{k+i|k},\hu_{k+i|k}) -
    \ell(\hat{\z}_{k+i|k},\hu_{k+i|k})} 
      + \norm{\Vf(\hx_{k+N|k}) - \Vf(\hat{\z}_{k+N|k})} \\
    &\leq \sum_{i=0}^{N-1} \sigma_\ell \circ\, \sigma_x^i(\norm{\xk-\zk})
      + \sigma_{\Vf} \circ\, \sigma_x^N(\norm{\xk-\zk})
      =: \sigma_V(\norm{\xk-\zk}) \ ,
  \end{align*}
  where $\circ$ denotes the concatenation of functions (\eg $\sigma_1
  \circ\, \sigma_2(x) = \sigma_1(\sigma_2(x))$) and $\sigma_V(\cdot)$ is
  a \K-function.
  Therefore the cost function is \uc in \xk for all $\xk \in \XNo$ and
  all \uS.
  
  As shown, for every $\xk \in \OOo$ the state constraints do not
  become active.
  Thus, the optimal solution $\uS^*$ of \eqref{eq:OCP} is feasible for
  every $\x_0 \in \OOo$ and we obtain
  \begin{align*}
    \norm{\VN^*(\xk|\Dk) - \VN^*(\zk|\Dk)} = \norm{\VN(\xk,\uS^*) -
    \VN(\zk,\uS^*)} \leq \sigma_V(\norm{\xk-\zk}) \ .
  \end{align*}
  Therefore, the value function $\VN^*(\xk|\Dk)$ is also \uc in \xk for
  all $\xk \in \OOo$ and a given data set \Dk.
  
  At last we show that the value function is a ISS-Lyapunov function.
  Since $\VN^*(\xk|\Dk)$ is a Lyapunov function for the nominal system
  (\cref{thm:nominal_stability}) there exists \Kinfty-functions
  $\alpha_1(\cdot), \alpha_2(\cdot), \alpha_3(\cdot)$, such that
  $\alpha_1(\norm{\xk}) \leq \VN^*(\xk|\Dk) \leq \alpha_2(\norm{\xk})$
  and $\VN^*(\xkp|\Dkp) - \VN^*(\xk|\Dk) \leq
  -\alpha_3(\norm{\xk})$.
  Moreover, from \eqref{eq:NARXstateModel} we have that $F(\xk,\uk,\ep)$
  is affine in \ep and is therefore \uc in \ep.
  Then, there exists a \K-function $\sigma_e(\cdot)$ such that $\norm{
    F(\xk,\uk,e_1) - F(\xk,\uk,e_2)
  } \leq \sigma_e(\abs{e_1-e_2})$ for all $\xk \in \XNo$, all $\uk
  \in \U$, and all $\abs{\ep} \leq \mu$.
  From these facts, it can be inferred that
  \begin{align*}
    \VN^*(\xkp|\Dkp) - \VN^*(\xk|\Dk) 
      &= \VN^*\big( F(\xk,\kMPC(\xk),\ep)|\Dkp \big) - \VN^*(\xk|\Dk) \\
      &= \VN^*\big( F(\xk,\kMPC(\xk),\ep)|\Dkp \big) -
        \VN^*\big( F(\xk,\kMPC(\xk),0)|\Dkp \big) \\
      & \quad + \VN^*\big( F(\xk,\kMPC(\xk),0)|\Dkp \big) -
        \VN^*(\xk|\Dk) \\
      &\leq \norm{\VN^*\big( F(\xk,\kMPC(\xk),\ep)|\Dkp \big) -
        \VN^*\big( F(\xk,\kMPC(\xk),0)|\Dkp \big)} - \alpha_3(\norm{\xk}) \\
      &\leq \sigma_V\big( \norm{F(\xk,\kMPC(\xk),\ep) -
        F(\xk,\kMPC(\xk),0)} \big) - \alpha_3(\norm{\xk}) \\
      &\leq \sigma_V \circ\, \sigma_e(\abs{\ep}) - \alpha_3(\norm{\xk})
      \ .
  \end{align*}
  Hence, $\VN^*(\xk|\Dk)$ is a ISS-Lyapunov function and the closed-loop
  system $\xkp = F\big(\xk,\kMPC(\xk|\Dk),\ep\big)$ is ISS w.r.t. \ep
  for all $\x_0 \in \OOo$.  
\end{proof}


\begin{remark}[(Differences in soft and hard output constraints)]
  In the case of soft constraints \Ys, the proposed controller ensures
  robust stability and constraint satisfaction for all initial states
  that lie in the feasible region \XNo of the optimal control problem.
  In the case of hard constraints \Yh, the proposed controller ensures
  robust stability and constraint satisfaction for all initial states
  in a robust feasible set $\OOo \subset \XNo$ where the constraints
  are not active.
  Thus, from a practical point of view, if in the soft constraints case
  the initial state \xk leads to a feasible solution, we then have
  $\xk \in \XNo$ and the above guarantees hold.
  If in the hard constraints case the initial state \xk leads to a
  feasible solution, then we also have $\xk \in \XNo$.
  However, in that case, one cannot be sure if also $\xk \in \OOo$ is
  satisfied.
  If $\xk \notin \OOo$, then feasibility might be lost at one point.
  Thus, for safety critical applications the set \OOo would 
  required to be known in order to check $\xk \in \OOo$, which is
  challenging because \OOo (as well as \XNo) can in general not be
  computed but has to be estimated via simulations.\cite{Rawlings2009}
  However, this issue could be circumvented if the hard constraints
  were tightened\cite{Limon2009}, thereby enlarging \OOo.
\end{remark}

\begin{remark}
  Notice that the ISS property is based on the uniform continuity of the
  optimal cost function and this does not depend on the size of the
  error signal.
  Hence, even if \ep is larger than $\mu$ for a short period of time in
  which we assume that the feasibility of the \OCP is not lost, \ie \xk
  remains in \XNk and ends in \OOk, then the closed-loop ISS property
  and constraint satisfaction will still hold.
\end{remark}

\begin{remark}[(Generalization)]
  Theorems~\ref{thm:nominal_stability} and \ref{thm:robust_stability}
  are independent of the control input dimension and also
  hold for general errors \e independent of the concrete
  structure of the state \xk, \ie whether \xk is a vector comprised of
  NARX states or of physical states.
  Thus, the theorems also include the multi-input multi-output case.
  In addition, as long as the presented assumptions are satisfied, in
  particular the update rule in Theorem~\ref{thm:nominal_stability},
  the stability results also hold for the case of
  online hyperparameter optimization and even further, for general
  prediction models $\FN(\xk,\uk|\Dk)$ that are updated online, \ie the
  stability guarantees are not confined to the use of GP prediction
  models.
\end{remark}

A necessary condition in Theorem~\ref{thm:robust_stability} is that the
nominal model $\FN(\xk,\uk|\Dk)$ is uniformly continuous in \xk for all
$\xk \in \XNo$, all $\uk \in \U$, and all \Dk during the prediction
horizon.
In the case of \GPs, this can be guaranteed by the following proposition.

\begin{proposition}[GP Uniform Continuity\cite{Maiworm2018b}]
  The nominal model \eqref{eq:predictionModel} is uniformly continuous
  in \xk if $\fN(\xk,\uk) = \postM(\wk|\D)$ is uniformly continuous in
  \xk.
  Since the prior mean $m(\wk)$ is added to $\postM(\wk|\D)$,
  the prior mean has to be \uc in \xk\footnote{The prior mean is usually
  specified by the user and often set to zero. Thus uniform continuity
  of $m(\w)$ is not an issue.}.
  Then, one way to ensure that $\postM(\wk|\D)$ is \uc in \xk, is to
  employ continuously differentiable kernels (e.g. the squared
  exponential covariance function, the Mat\'{e}rn class covariance
  function with appropriate hyperparameters, or the rational quadratic
  covariance function).
  In that case the process is mean square differentiable
  \cite{Abrahamsen1997, Rasmussen2006}, \ie the posterior mean function
  is differentiable and therefore also uniformly
  continuous\footnote{Continuous differentiability is a stronger
  assumption than uniform continuity.}.
  \label{prop:GP_uniform_continuity}
\end{proposition}

\begin{remark}
  Although not required for \cref{thm:robust_stability}, note that
  uniform continuity of the process $F(\xk,\uk,\ep) = \FN(\xk, \uk |
  \Dk) + \VEC d \ep$ in \xk is ensured if $\FN(\xk,\uk|\Dk)$ is \uc in
  \xk, which can be established via \cref{prop:GP_uniform_continuity}.
\end{remark}

\paragraph{Resulting Prediction Errors}

We finish this section with a discussion on the prediction error $\ep =
\yk - \hyk$.
According to Theorem~\ref{thm:robust_stability}, the smaller the error
bound $\abs{\ep} \leq \mu$, the larger the feasible set \OOo.
Since the noise $\epsilon$ (affecting \yk) is in practice bounded by a
finite $\bar\epsilon$, the
error bound $\mu$ is finite if \hyk is finite (given of course that the
original process \yk is finite), which translates to the necessity that
the GP posterior mean \eqref{eq:postM} is bounded.

From a theoretical point of view, such a bound exists under certain
conditions.
Note that the posterior mean (with zero prior mean $m(\w)=0$) can also
be expressed via $\postM(\w^*|\D) = \sum_{i=1}^n \alpha_i
k(\w_i,\w^*)$, with $\w_i \in \wb$, as a linear combination of $n$
kernel functions\cite{Rasmussen2006} that determines a
\emph{reproducing kernel Hilbert space} (RKHS).
As shown in Steinwart and Christmann\cite{Steinwart2008}, a bound
in the RKHS exists if \emph{universal} kernels are employed.
One such kernel is, for instance, the squared
exponential covariance function\footnote{The
squared exponential covariance function is sometimes also denoted as
\emph{Gaussian radial basis function}.
Especially in the field of neural networks or support vector
machines.
}
\eqref{eq:covSEard} for which the existence of a bound had already
been shown by Park and Sandberg\cite{Park1991}. 
De Nicolao and Pillonetto\cite{Nicolao2008} have presented a very
similar result when modeling the impulse response via a spline kernel.
The result has also been used in Pillonetto and
Chiuso\cite{Pillonetto2009}.
Furthermore, Engel\cite{Engel2005} and Srinivas et
al\cite{Srinivas2012} provide ways to explicitly compute the bound,
though only with high probability.
%

In practice, however, $\postM(\w^*|\D)$ will generally be bounded
assuming that the employed GP prior is well chosen and sufficiently
informative training data \D is used.
Thus, the actual bound depends on the designer's choices regarding the
particular employed GP model and the involved tuning parameters.
Among these, in particular the thresholds for the prediction error and
posterior variance for the presented rGP approach.

\section{Simulations}
\label{sec:simulations}

In this section, we provide simulation results for the presented rGP-MPC
scheme and consider a continuous stirred-tank reactor as simulation
case study.
We present the model equations, the training data set generation, and 
the terminal components for the MPC based on the linearized GP
posterior mean function.
The closed-loop simulations involve investigations regarding the tuning
parameters of the rGP-MPC, the influence of different initial training
data sets, as well as comparisons with other MPC controllers.

\subsection{Continuous Stirred-tank Reactor}

As exemplary case study, we consider the continuous stirred-tank reactor
(CSTR), where a substrate $A$ is converted into product $B$.
\cite{Seborg1989}
The following set of differential equations describes the
reactor dynamics:
\begin{subequations}
\begin{align}
  \dot C_A(t) &= \frac{q_0}{V} \big( C_{A \text f} - C_A(t) \big)
                 - k_0 \exp \left( \frac{-E}{R T(t)} \right) C_A(t) \\
  \dot T(t) &= \frac{q_0}{V} \big( T_\text f - T(t) \big)
               - \frac{\Delta H_\text r k_0}{\rho C_\text p}
               \exp\left( \frac{-E}{R T(t)} \right) C_A(t) \nonumber \\
            &\quad + \frac{U A}{V \rho C_p} \big( T_\text c(t) - T(t)
            \big) \\
  \dot T_\text c(t) &= \frac{T_\text{r}(t) - T_\text c(t)}{\tau} 
\end{align}
  \label{eq:CSTR}
\end{subequations}
The coolant temperature reference \unit[$T_\text{r}$]{(K)} is the
input and the concentration \unit[$C_A$]{(mol/l)} the
output, \ie $u = T_\text{r}$ and $y = C_A$.
The tank and coolant temperatures are $T$ and $T_\text c$, respectively.
The model parameters are given in Tab.~\ref{tab:CSTRparameters}.

\subsection{Training Data Sets}

A raw data set \Draw (depicted in Fig.~\ref{fig:dataSet}) is generated
using the plant \eqref{eq:CSTR}.
The data points $(z_i,\w_i)$ consist of values of $(\ykp, \yk, \ldots,
y_{k-\My},$ $u_k, \ldots, u_{k-\Mu})$, where $z = \ykp$ is going be the
GP output and $ \w = (\yk, \ldots, y_{k-\My},$ $u_k, \ldots, u_{k-\Mu})$
its corresponding regressor.
Based on this data, we generate the three training data sets \DO, \Dref,
and \Dcomb.
The set \DO is a local subset around the initial point $y_0 =
C_A = \unit[0.6]{mol/l}$. 
The associated input is $u_0 = T_\text{r} = \unit[353.5]{K}$.
The set \Dref is a local subset around the target reference point
$\yref = C_A = \unit[0.439]{mol/l}$ with associated input
$u_\text{ref} = T_\text{r} = \unit[356]{K}$.
The set $\Dcomb = \DO \cup \Dref$ is the union of the two sets.

\begin{figure}[htpb]
  \centering
  \includegraphics[width=0.6\linewidth]{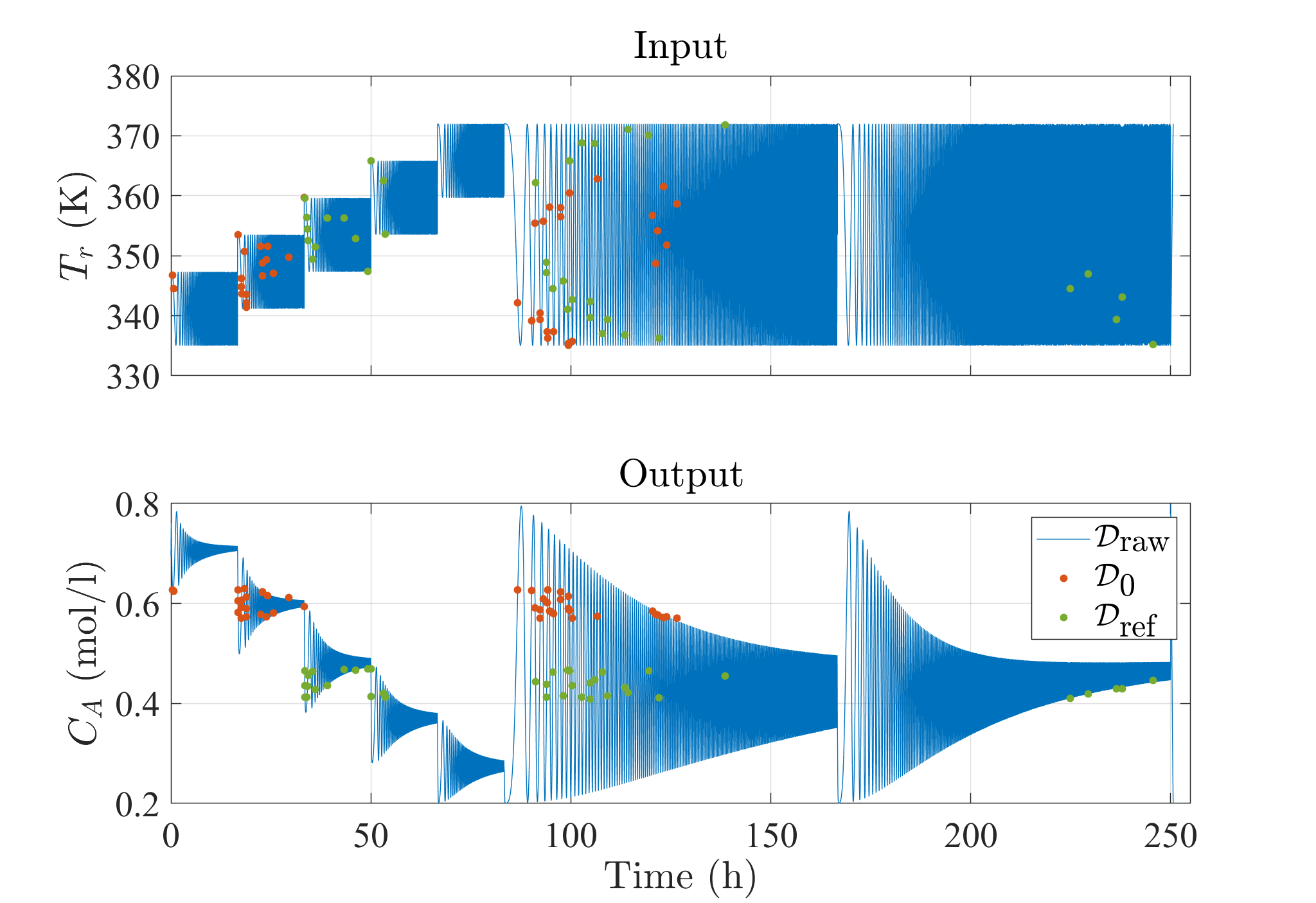}
  \caption{Training data sets:
  The raw data set \Draw was generated by chirp signals on the input. 
  The sets \DO and \Dref are local neighborhoods of the initial
  point $y_0$ and the reference point \yref and their associated inputs.
  }
  \label{fig:dataSet}
\end{figure}

The sets \DO and \Dref are generated by selecting first all points
$z = \ykp$ (and their corresponding \w)
that are located within a local neighborhood of the respective
set-points and second, by reducing the number of points via exclusion
of those that add only little information.
For a given data point $(z_i,\w_i)$, all following $(z_j,\w_j), j > i$,
are removed, for which $\norm{\w_i-\w_j} < \bar w$ with a chosen
threshold $\bar w$. 
As a result, the sets are less dense but still contain enough
informative data points.
The thresholds for \DO and \Dref are chosen such that both sets contain
approximately 40 data points.

\begin{remark}
  All input and output values are given in the original units of the
  system \eqref{eq:CSTR}.
  However, it is beneficial for the modeling process with the GP
  to normalize the input-output data to the interval $[0, 1]$.
\end{remark}

\subsection{GP Prediction Model}

For the GP prior we employ a constant mean function with constant
$c$.
Since the underlying process equations are smooth and to obtain the
universal approximation property (see
Section~\ref{sec:robust_stability}) we employ the covariance function
\eqref{eq:covSEard} with regressor $\w = [\yk, y_{k-1}, y_{k-2}, u_k]$.
According to \eqref{eq:NARXstate}, the NARX state is then $\xk = [\yk,
y_{k-1}, y_{k-2}]\T$.
The hyperparameters are $\thetaa = \{ c, l_1, l_2, l_3, l_4, \Sf \}$ and
are computed offline via maximization
of \eqref{eq:loglike} for each of the three data sets \DO, \Dref, and
\Dcomb.
We obtain three sets of hyperparameters respectively
(Tab.~\ref{tab:hyperparameters}) and with that three different GP
prediction models that use the same prior but different training data
sets and hyperparameters.
The cross validation results of these different GP models are shown in
Fig.~\ref{fig:GP_validation_error}, where we select test points
throughout the regions of the respective training data sets.
Test points are chosen such that they are not part of \DO,
\Dref, or \Dcomb.
As can be seen, appropriate GP predictions are achieved with prediction
error $\ep < \epU = \unitfrac[0.02]{mol}{l}$ and posterior standard
deviation $\sigma_+ < \sigmaU = \unitfrac[5 \cdot 10^{-3}]{mol}{l}$ for
all three GPs.

\begin{figure}[htpb]
  \centering
  \includegraphics[width=0.6\linewidth]{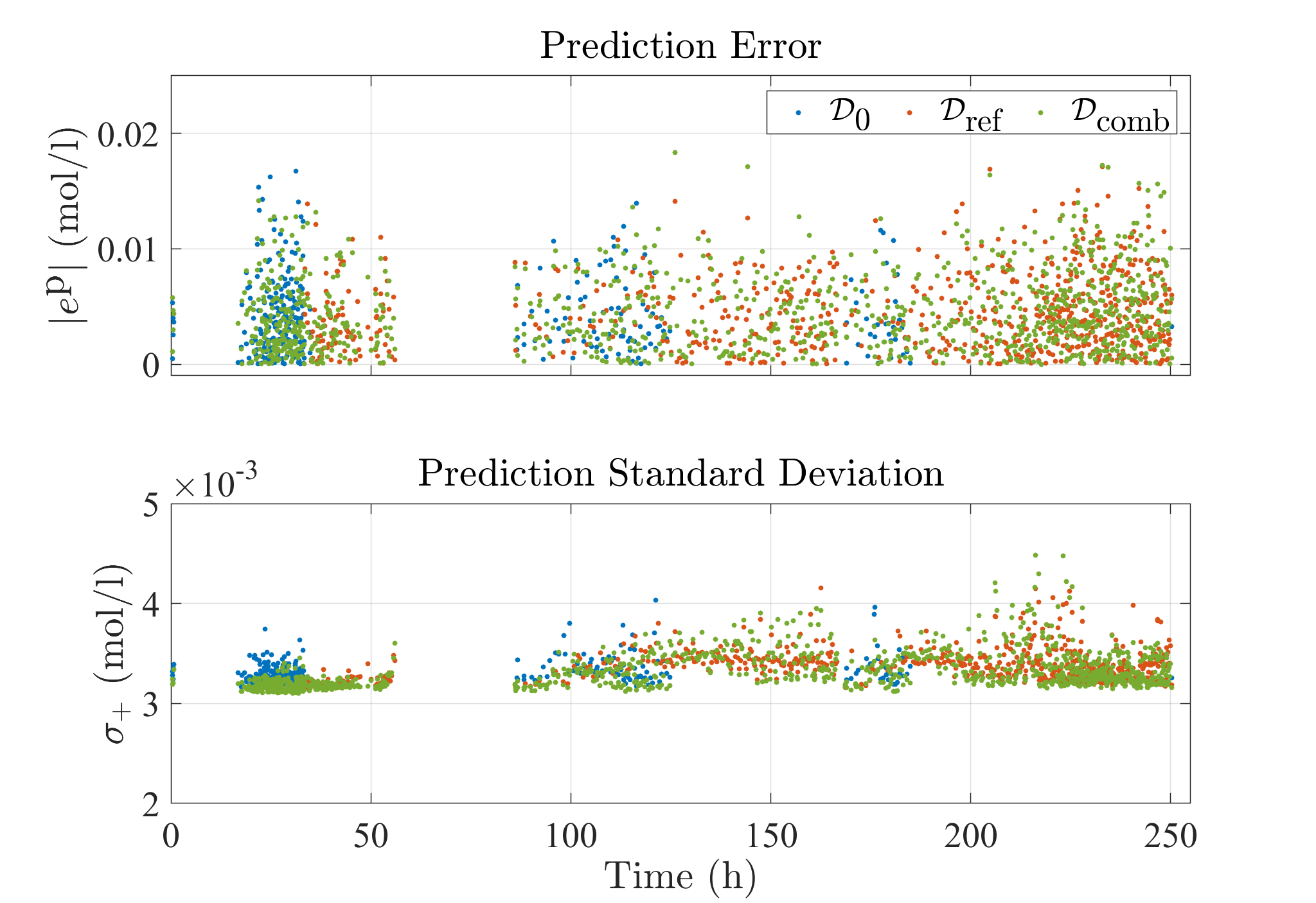}
  \caption{Cross validation results: Top, the prediction error \ep
  \eqref{eq:ep} is depicted.
  Bottom, the posterior standard deviation $\sigma_+(\w) =
  \sqrt{\postV(\w)}$.
  }
  \label{fig:GP_validation_error}
\end{figure}

\subsection{Optimal Control Problem}

The continuous-time model \eqref{eq:CSTR} is discretized with Euler's method
and a sampling time of $T_\text s = \unit[0.5]{min}$.
The input constraints are $\U = \{ \unit[335]{K} \leq T_\text r \leq
\unit[372]{K} \}$, the output constraints $\Yh = \{ \unit[0.35]{mol/L}
\leq C_A \leq \unit[0.65]{mol/l} \}$.
We add measurement noise
$\epsilon \sim \CAL N(0,\Sn)$ to the output
data with $\Sn = 0.003^2$, which we furthermore bound\footnote{
  According to the considered system class we add Gaussian noise with
  bounded support.
  Since \GPs are based on Gaussian noise with unbounded support, there
  is a small difference, which could be accounted for by GP
  warping\cite{Snelson2004}.
  However, due to the large bound of four standard deviations, the
  difference is so small that the following simulation results are equal
  to the ones with unbounded noise.
} 
by $\pm 4\sigma_\text{n}$.
The employed quadratic stage cost is given by
\begin{align*}
  \ls(\xk,u_k) = \norm{\xk - \xref}^2_{\Q} + \norm{u_k - u_\text{ref}}^2_R
\end{align*}
with $\Q = \text{diag}(100,0,0)$, and $R = 5$.
The prediction horizon is set to $N = 5$.
The resulting \OCP is solved in MATLAB using \texttt{fmincon}.

\subsection{Terminal Controller and Cost Function}
\label{sec:terminal_ingredients}

The terminal controller $\kf(\cdot)$ and cost function $\Vf(\cdot)$ can
be determined arbitrarily, as long as the assumptions in
Section~\ref{sec:stability} are satisfied.
We choose the terminal controller as $\kf(\x) = \VEC k\T(\x-\xref)
+ u_\text{ref}$ and the terminal cost function as $\Vf(\x) =
\norm{\x-\xref}^2_{\P}$, where $\k \in \R^3$ and $\P \in \R^{3\times3}$
are computed using the linearization of the prediction model
\eqref{eq:predictionModel} with the GP model, based on the training data
set \Dref obtained near the reference \xref.

The linearization of the nominal NARX model $\xkp = \FN(\xk,u_k)$ with
$\xk = [\yk, y_{k-1}, y_{k-2}]$ takes the form
\begin{align*}
  \Matrix{ y_{k+1} \\ y_k \\ y_{k-1} }
  =
  \underbrace{
    \Matrix{ a_{11} & a_{12} & a_{13} \\
              1     &    0   &   0    \\
              0     &    1   &   0 
           }
         }_{\A}
  \Matrix{ y_k \\ y_{k-1} \\ y_{k-2} }
  +
  \underbrace{
    \Matrix{ b_1 \\ 0 \\ 0}
  }_{\b}
  u_k \ .
\end{align*}
As the next output is computed using the GP, \ie $\ykp =
\postM(\wk)$, the parameters $a_{11}, a_{12}, a_{13}, b_1$ can be
determined using the posterior mean gradient derived in the appendix,
Sec.~\ref{sec:posterior_mean_gradient}.
In particular we have $[a_{11}, a_{12}, a_{13}, b_1] = \nabla
\postM(\w_\text{ref})\T$ with $\w_\text{ref} = [\yref, \yref, \yref,
u_\text{ref}]$.
The resulting linear model becomes
\begin{align}
  \xkp = 
    \Matrix{
      0.162 & 0.005 & -0.012 \\
      1     & 0     & 0 \\
      0     & 1     & 0 \\
    }
    \xk + 
    \Matrix{ -0.034 \\ 0 \\ 0
    }
    u_k \ .
  \label{eq:linearModel}
\end{align}

We define the feedback vector as $\k = \P \s$ with $\s \in \R^3$, $\P =
\P\T > 0$, and $\G = \P\inv$.
We furthermore define the state constraint set $\X = \Yh \times \Yh \times
\Yh$ and reformulate \X and \U as polyhedral sets of the form
$\X = \left\{ \x \in \R^\Nx : \q_i\T \x \leq r_i, \ i = 1,\ldots,
n_\X \right\}$ and $\U = \left\{ u \in \R^\Nu : v_l\T u \leq t_l, \ l =
1, \ldots, n_\U \right\}$, where $n_\X$ and $n_\U$ are
the respective numbers of inequalities.
Then, we compute \s and \P offline by solving the semidefinite
optimization problem\cite{Maiworm2015}
\begin{align}
\begin{aligned}
  \underset{\G,\s}{\max} ~& \log\big(\det (\G) \big) \\
    \text{s.t.~} &\G = \G\T > 0  \\
        & \Matrix{\G & \left( \A \G + \b \s\T \right)\T \\
                      \left( \A \G + \b \s\T \right) & \G
                  } \geq 0 \\
        & \Matrix{\G         & (\G \q_i)  \\
                  (\G \q_i)\T & r_i^2
                  } 
    \geq 0, \quad\forall i \in \left\{ 1,\dotsc, n_\X \right\} \\ 
        & \Matrix{\G           & (\s v_l)  \\
                  (\s v_l)\T & t_l^2
                  } 
    \geq 0, \quad\forall l \in \left\{ 1,\dotsc, n_\U \right\} 
\label{eq:SDP}
\end{aligned}
\end{align}
and obtain
\begin{align*}
  \k\T = 
    \Matrix{ 1.745 & 0.082 & -0.001 } 
   \ \text{ and } \ 
  \P =
    \Matrix{
      16.38  & -0.556 & -0.066 \\
      -0.556 & 16.32  & -0.554 \\
      -0.066 & -0.554 & 16.30 \\
    }
    \ .
\end{align*}

The optimization problem \eqref{eq:SDP}\cite{Maiworm2015} results
from using the Schur complement in combination with the discrete-time
Lyapunov equation and the support function concept of
closed convex sets.
The resulting \s and \P are such that the closed-loop linearized system
is asymptotically stable in $\Xf = \left\{ \x \in \R^\Nx : \Vf(\x) =
\norm{\x-\xref}^2_{\P} \leq 1 \right\} \subseteq \X$ and $\k \Xf
\subseteq \U$.

\begin{remark}
  It has been proven in the literature that the quadratic Lyapunov
  function holds for the nonlinear system in a certain neighborhood of
  the equilibrium point.
  The terminal region definition $\Xf = \left\{ \x \in
  \R^\Nx : \Vf(\x) = \norm{\x-\xref}^2_{\P} \leq \nu \right\}$,
  parameterized with $\nu$, could be used to characterize this
  neighborhood.
  Then we would need to take the nonlinear remainder term into account
  to calculate a particular value for $\nu$, which would require
  the solution of a global optimization problem.
  Such a problem could be solved by using scenarios or a Monte
  Carlo approach.
  However, since the optimal control problem does not need the
  terminal region constraint, $\nu$ is not required.
\end{remark}

\subsection{Simulation Results}

First, we simulate the set-point change from $(u_0,y_0)$ to
$(u_\text{ref},\yref)$ and compare the closed-loop results of the
rGP-MPC, a batch GP approach (bGP-MPC) that uses a fixed training data
set, and an output feedback MPC scheme (oMPC) that uses the
model equations \eqref{eq:CSTR} and acts as a performance bound.
We evaluate the performance for the three cases, where \DO, \Dref,
and \Dcomb are used as initial training data sets.
The bGP and rGP are initialized with the same initial training data 
and hyperparameters but the rGP updates its training data set during
operation.
We set $\epU = \sigmaU = 0$ such that every data point is considered
as a candidate for inclusion\footnote{Not every data point is added
  due to the update rule of Theorem~\ref{thm:nominal_stability}.} 
with no upper limit on the number of data points.
Hence, no points are removed.
Due to the stochastic nature of the noise component, we simulate
each case $\Nsim = 50$ times.
The results are depicted in Fig.~\ref{fig:MPC_comparison_D0} to
Fig.~\ref{fig:MPC_comparison_Dcomb}.
To quantify the performance we employ the measure
\begin{align}
  \bar V = \frac{1}{\Nsim} \sum_{j=0}^{\Nsim} \sum_{k=0}^{\Nstep}
    \ell \big( \x_k^j, \uk^j \big) \,
  \label{eq:performance}
\end{align}
which averages the stage costs of the resulting state and input sequences
over all time steps $k \in \{0, 1, \ldots, \Nstep\}$, as well
as the individual simulations $j \in \{1, 2, \ldots, \Nsim\}$.
The resulting $\bar V$ values are presented in Tab.~\ref{tab:performance}.

\begin{figure}[htpb]
  \centering
  \includegraphics[width=0.6\linewidth]{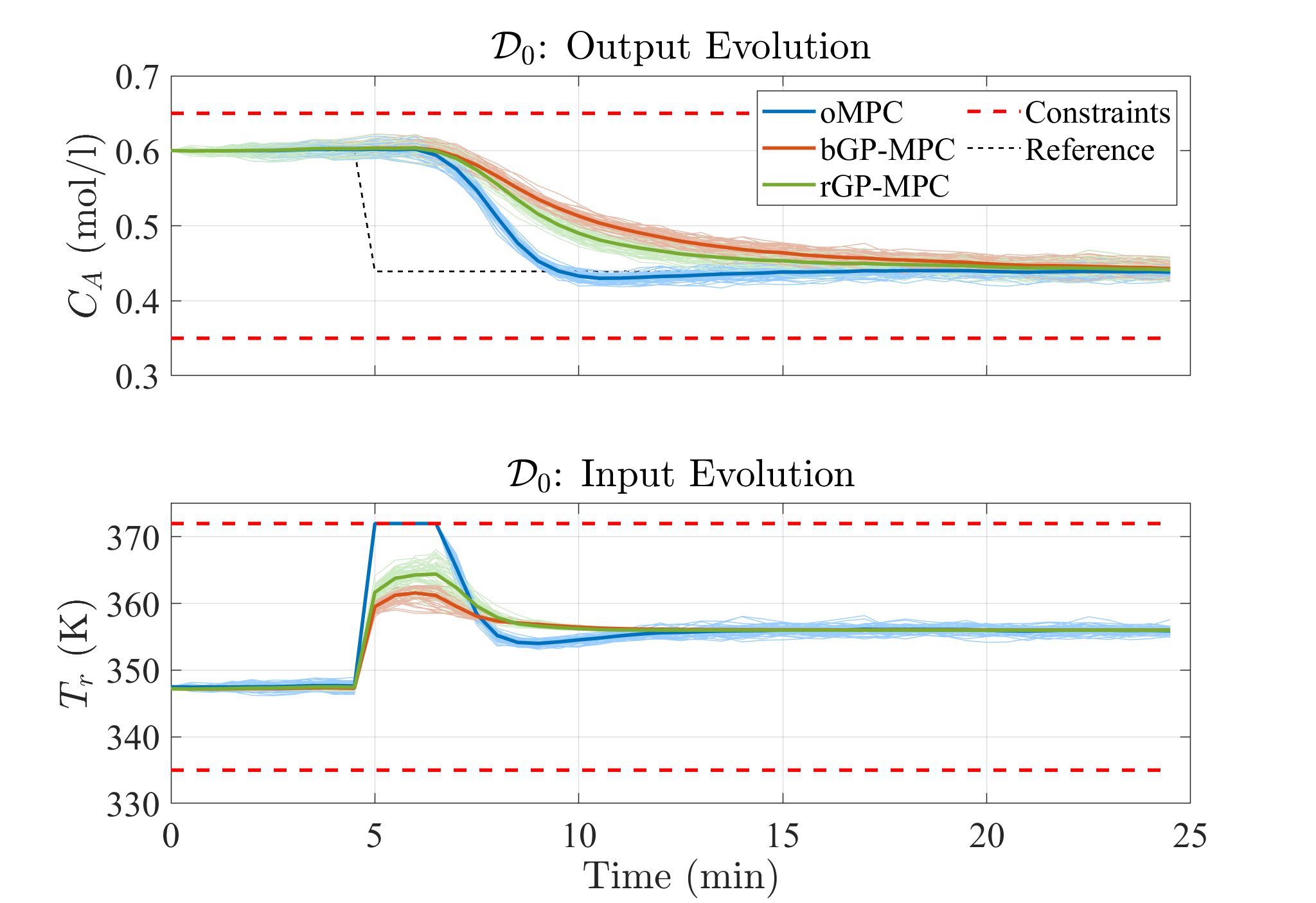}
  \caption{Comparison of the three MPC schemes for the case of initial
  training data \DO.
  Thin lines represent individual simulations, thick lines represent
  mean values.
  }
  \label{fig:MPC_comparison_D0}
\end{figure}

\begin{figure}[htpb]
  \centering
  \includegraphics[width=0.6\linewidth]{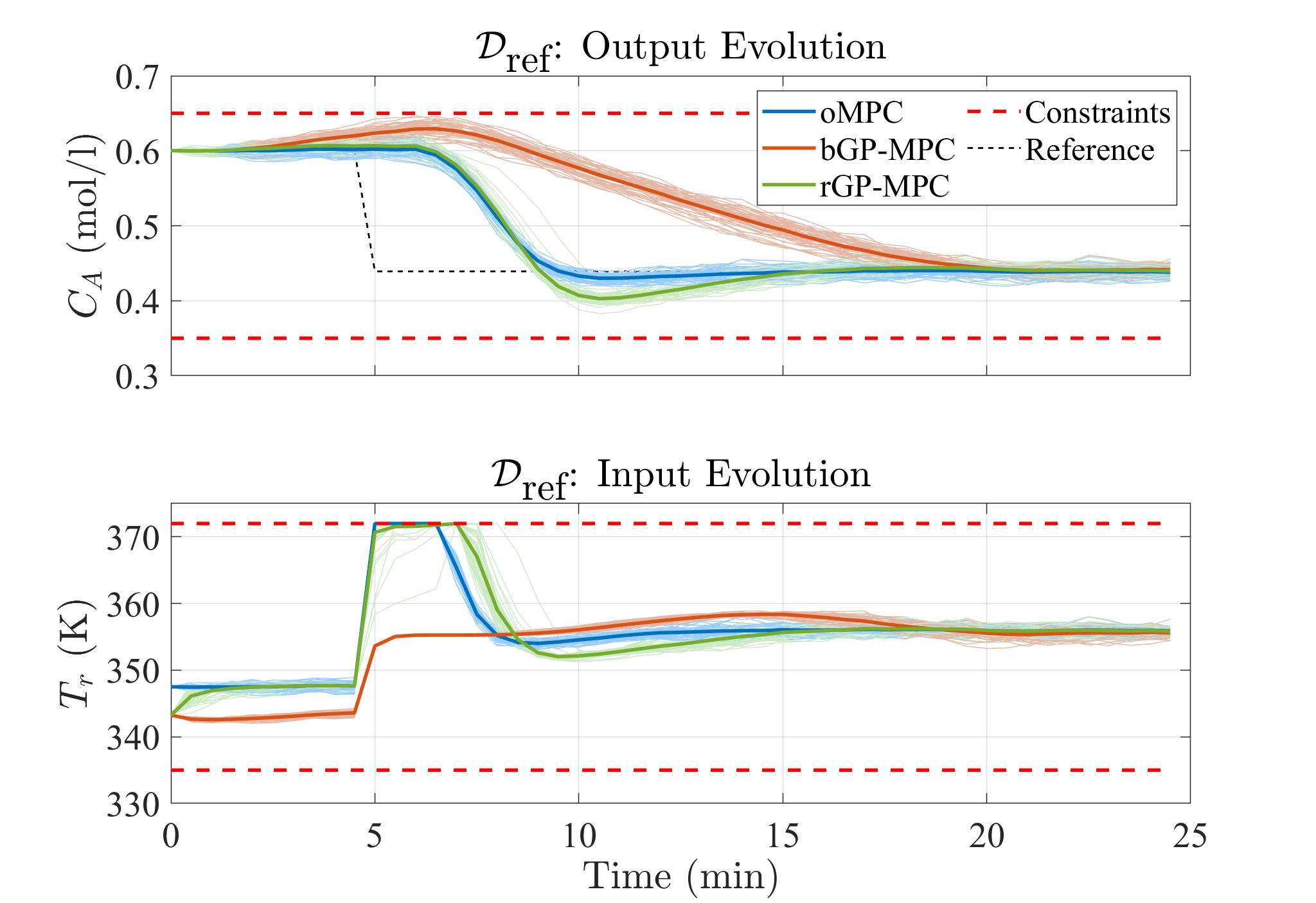}
  \caption{Comparison of the three MPC schemes for the case of initial
  training data \Dref.
  Thin lines represent individual simulations, thick lines represent
  mean values.
  }
  \label{fig:MPC_comparison_Dref}
\end{figure}

\begin{figure}[htpb]
  \centering
  \includegraphics[width=0.6\linewidth]{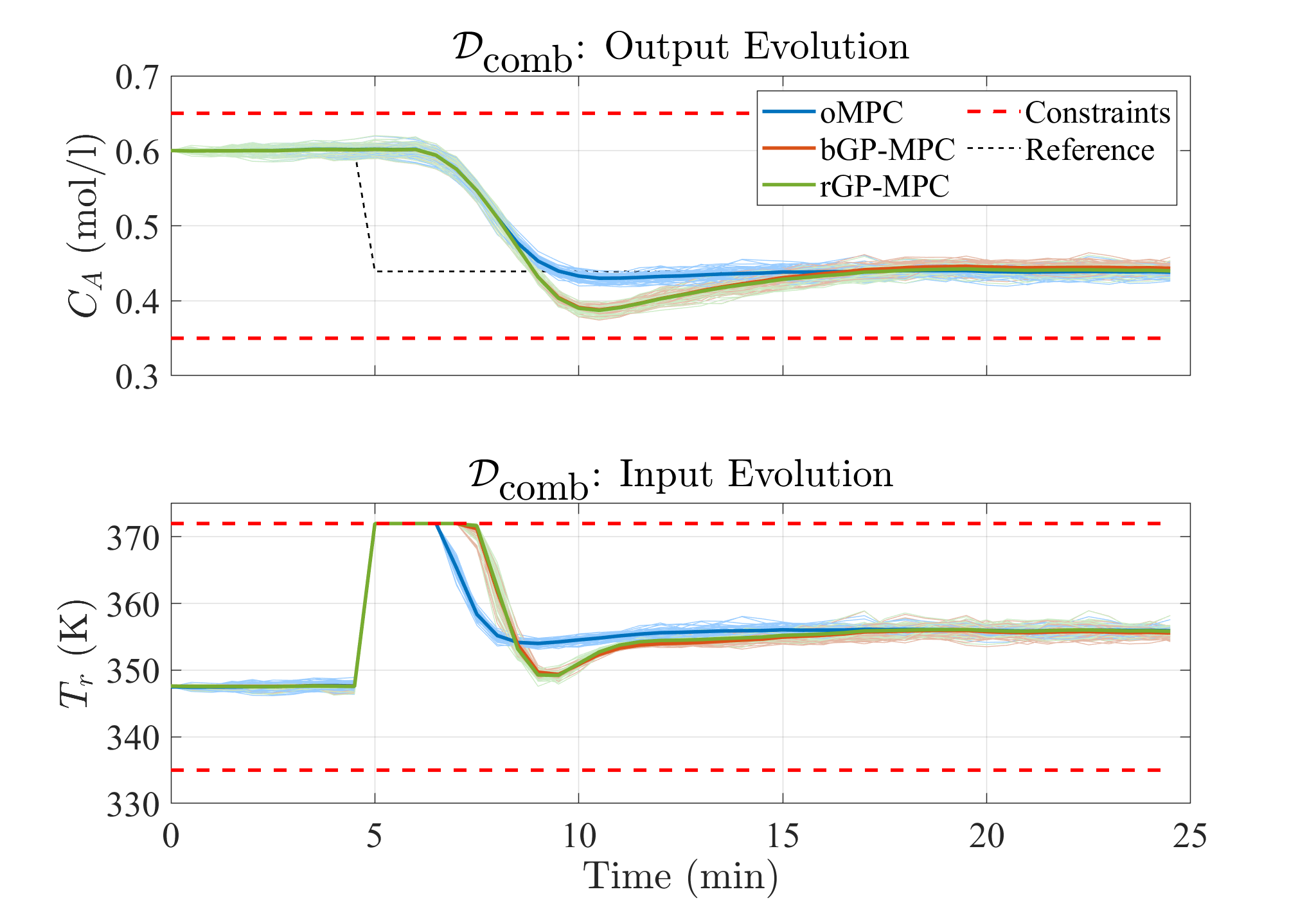}
  \caption{Comparison of the three MPC schemes for the case of initial
  training data \Dcomb.
  Thin lines represent individual simulations, thick lines represent
  mean values.
  }
  \label{fig:MPC_comparison_Dcomb}
\end{figure}

As expected, the oMPC scheme that uses the true model performs
best and always the same (see \cref{tab:CSTRparameters}) because it does
not depend on any training data points.
The rGP outperforms the bGP in the \DO and \Dref cases due to the 
additional information gained during operation.
The performance difference is especially large for \Dref, where
the bGP, throughout the whole operation, has only data points at the
reference at its disposal but not at the initial condition.
The rGP performs significantly better due to the added data points at
the beginning of operation.
In the \Dcomb case, the rGP and bGP performance is almost the same for
the employed training data points.

\begin{remark}
  The previous simulation results suggest that one should in general
  prefer the \Dref case over the other cases, which is convenient for
  the used MPC scheme because knowledge at the reference is required
  anyway to determine the terminal cost and controller.
  Furthermore, this also suggests a practical rule for offline
  hyperparameter determination, namely that the hyperparameters should
  be optimized for a data set that contains the target reference.
\end{remark}

In the second set of simulations, we investigate the influence of
different thresholds used in Rule~\ref{rule:add_data_point}, \ie
different values for the maximum prediction error \epU and the maximum
prediction variance \sigmaU.
To this end, we start with Fig.~\ref{fig:evolution_error_variance} that
combines the rGP results of the previous figures for the three training
data cases, together with the now plotted evolution of the prediction
error \ep and the prediction variance $\sigma_+^2$.
In particular the prediction variance illustrates nicely the difference
between the three cases.
In the case of \DO, the variance is small at the beginning and
increases around $t = \unit[8]{min}$ when the system leaves the
neighborhood of the initial condition and moves towards the reference.
The same holds, but the other way round, for the case with \Dref, where
the initial ($t < \unit[3]{min}$) large error and variance is caused by
their computation before the first data points are added to the training
set.
The prediction error bound $\mu$ is 0.033, 0.021, and 0.024 for the
cases \DO, \Dref, and \Dcomb respectively.

\begin{figure}[htb]
  \centering
  \includegraphics[width=0.6\linewidth]{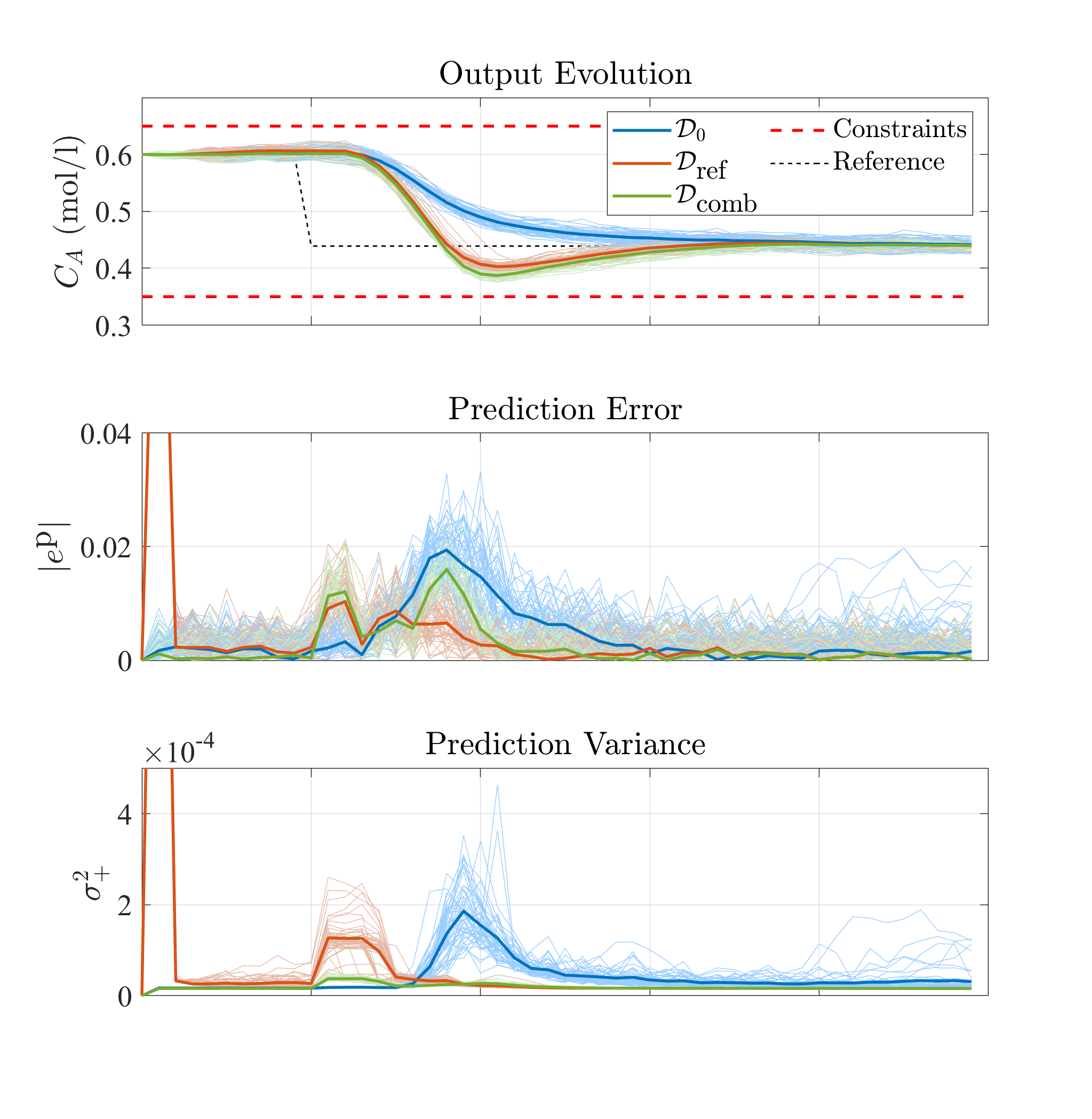}
  \caption{Simulation results with the rGP-MPC for the different training
  data cases together with the absolute value of prediction error
  $\abs{\ep}$ and the prediction variance \postV. 
  }
  \label{fig:evolution_error_variance}
\end{figure}

Fig.~\ref{fig:sim_peThreshold_D_ref} and
Fig.~\ref{fig:sim_varThreshold_D_ref} show results for different
threshold values, where we focus for the sake of brevity on the
simulation case with \Dref.
The results illustrate that instead of adding all data points, almost
the same closed-loop performance can be achieved by adding only a
fraction of them.
Hence, this shows not only that online learning can be achieved but also
that it allows working with significantly smaller training data sets,
which in turn result in lower computational costs.

\begin{figure}[htpb]
  \centering
  \includegraphics[width=0.6\linewidth]{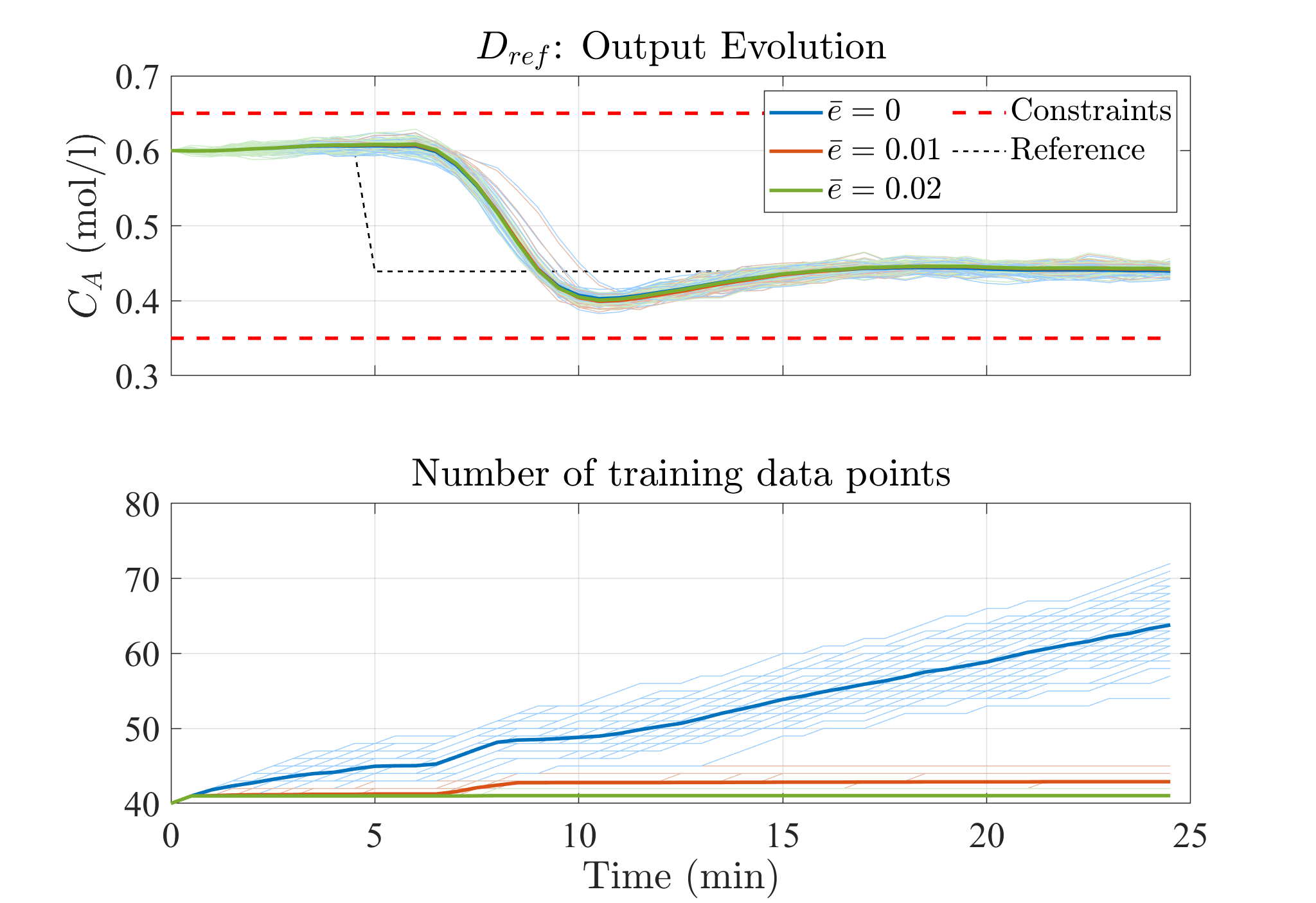}
  \caption{Influence of \epU on the rGP-MPC with initial training data \Dref.
  With $\epU = 0$, every encountered data point is considered to be
  added to the training data set.
  The variance threshold \sigmaU was set to a large value to not
  affect the result.}
  \label{fig:sim_peThreshold_D_ref}
\end{figure}

\begin{figure}[htpb]
  \centering
  \includegraphics[width=0.6\linewidth]{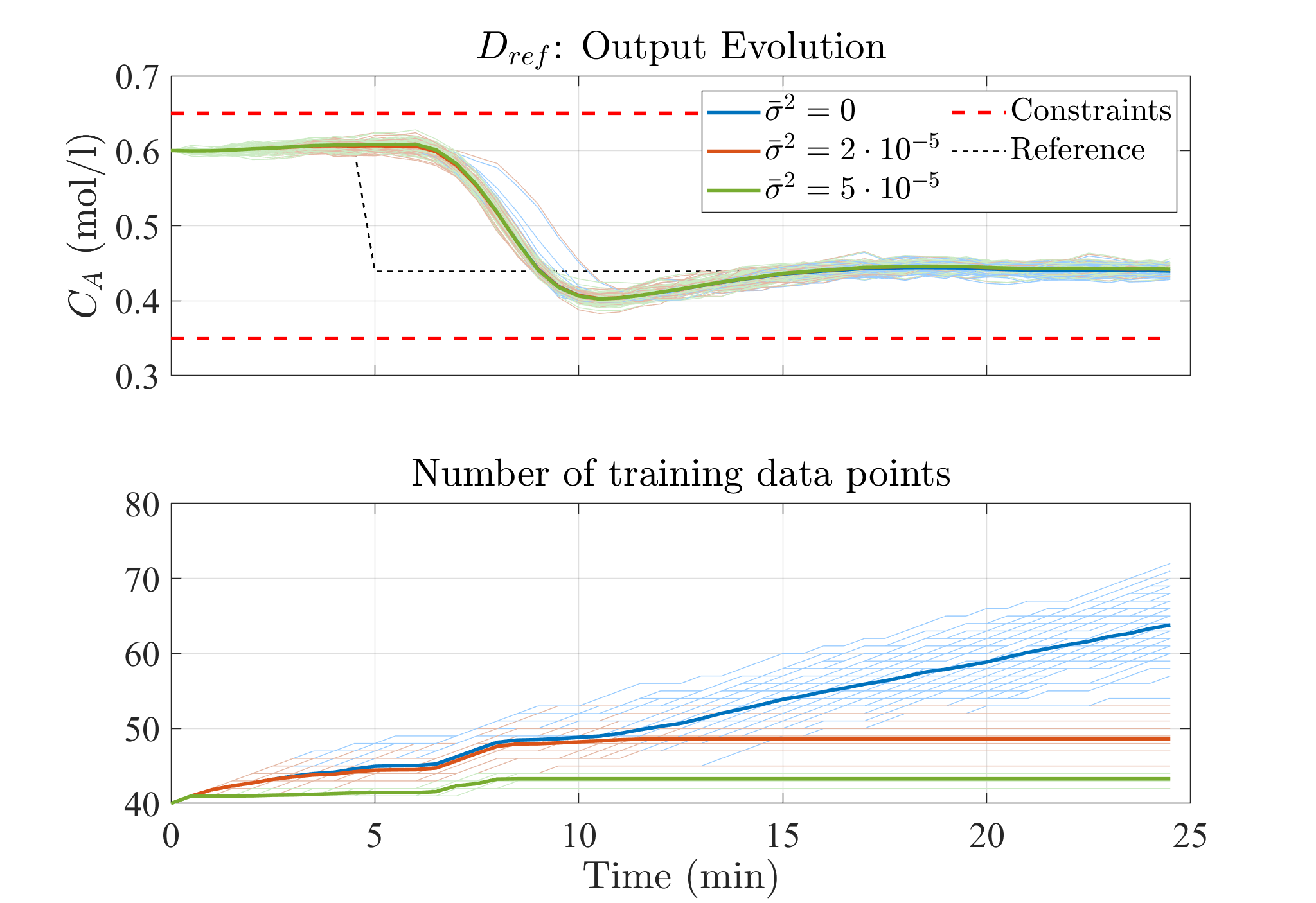}
  \caption{Influence of \sigmaU on the rGP-MPC with initial training
  data \Dref.
  With $\sigmaU = 0$, every encountered data point is considered to be
  added to the training data set.
  The prediction error threshold \epU was set to a large value to not
  affect the result.
  }
  \label{fig:sim_varThreshold_D_ref}
\end{figure}

After evaluating the influence of the parameters of
Rule~\ref{rule:add_data_point}, we illustrate the influence of the
update rule in \cref{thm:nominal_stability}, which guarantees a decrease
of the value function.
To this end, we continue with the \Dref case and additionally insert
outliers into the output measurements in the course of the simulations.
The effect of the update rule is shown in \cref{fig:sim_7_updateRule}.
With it, the results are almost the same as before, except for the
distortions due to the outliers, which however are compensated shortly
after.
All simulation outcomes are very similar in that case.
Without the update rule, the resulting mean output sequence
is different but not necessarily worse (smaller rise time, similar
settling time, no overshoot) than the mean output sequence with the
update rule.
Some of the individual simulation outcomes perform even better,
which is an indication that data points with valuable information are
indeed discarded by the update rule as was also pointed out in
\cref{rem:properties_update_rule}.
On the other hand, the variability among the individual simulations is
much larger.
Several of the simulated output evolutions converge slower
to the target and some do not converge at all until the end of the
simulation.
This is a direct result of the corresponding input sequences computed
by the optimizer.
In between \unit[5]{min} and \unit[11]{min}, the deviation of the mean
input sequence from the optimal input sequence of the performance
bound (oMPC\cf \cref{fig:MPC_comparison_D0,} to
\cref{fig:MPC_comparison_Dcomb}) is larger than in the case with the
update rule.
Furthermore, the individual input sequence outcomes vary
considerably, even hitting the lower constraints.
Due to the inclusion of every encountered data point candidate, the
prediction model changes in some cases in an unfavorable way during the
respective simulations, which leads to the depicted results.
Note that qualitatively the same results (including not converging
output sequences) are obtained, even without outliers.
For instance, between the reference change at \unit[5]{min} and the
first outlier at \unit[7]{min}, we observe that the input sequences
already deviate considerably from the case with an active update rule,
\ie the outlier is not the cause but usual noisy data points.
This illustrates the importance of the update rule in
\cref{thm:nominal_stability}, not only for theoretical guarantees but
also in terms of practical application.

\begin{figure}[htpb]
  \centering
  \includegraphics[width=0.6\linewidth]{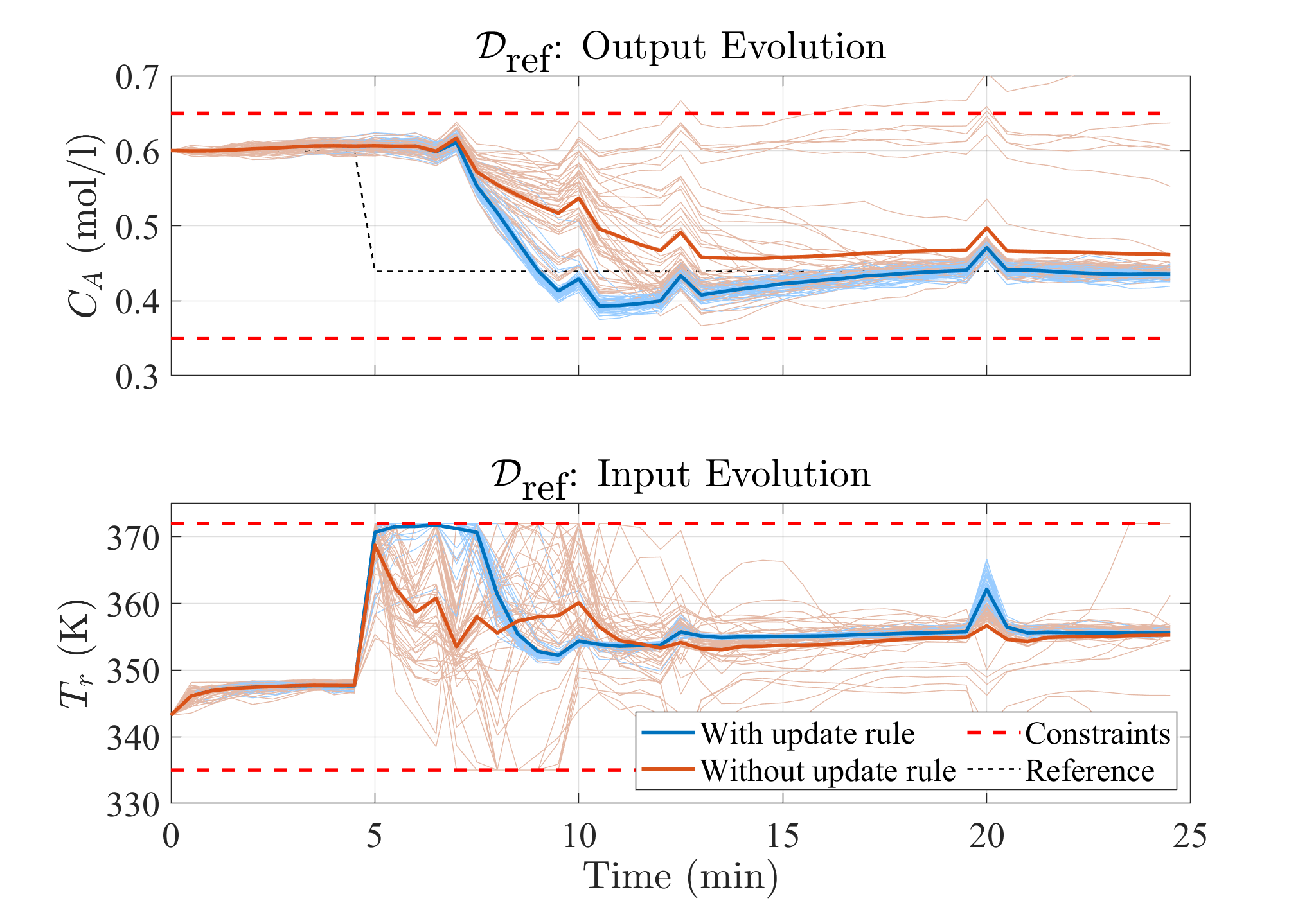}
  \caption{Influence of the update rule in \cref{thm:nominal_stability},
  which permits inclusion of data point candidates only if they result
  in a decreasing value function.
  Outliers are generated at \unit[7]{min}, \unit[10]{min},
  \unit[12.5]{min}, and \unit[20]{min}.
  }
  \label{fig:sim_7_updateRule}
\end{figure}

Next, we consider the case that the number of training data points is
limited by $M$.
For the case of \Dref we set $M = 40$, which is the number of initially
available training points, \ie the training data set cannot increase but
old data points are exchanged with newer more informative ones.
To this end, whenever a new point is added, the oldest data point is
removed.
In Fig.~\ref{fig:limitedTrainPoints} we compare the bGP
(the initial training data set is not updated at all), the rGP with $M =
\infty, \epU = \sigmaU = 0$ (every encountered data point is considered
to be added), and the rGP with $M = 40, \epU = 0.01, \sigmaU =
2\cdot10^{-5}$ (data points are only exchanged).
The bGP result is the same as in Fig.~\ref{fig:MPC_comparison_Dref} and
represents the worst case because the training data set is not updated
at all.
The $M = \infty$ case on the other hand represents the performance bound
for this specific case because it includes the maximum of the incoming
data points and does not remove any.
As can be seen, the reaction of the limited case is a bit slower than
the performance bound case but the resulting settling times are almost
identical.
Thus, with a training data set of only 40 points, where the points are
exchanged during operation, almost the same performance can be achieved
for the considered example as if every encountered point was included
in the training data set \D.

\begin{figure}[htpb]
  \centering
  \includegraphics[width=0.6\linewidth]{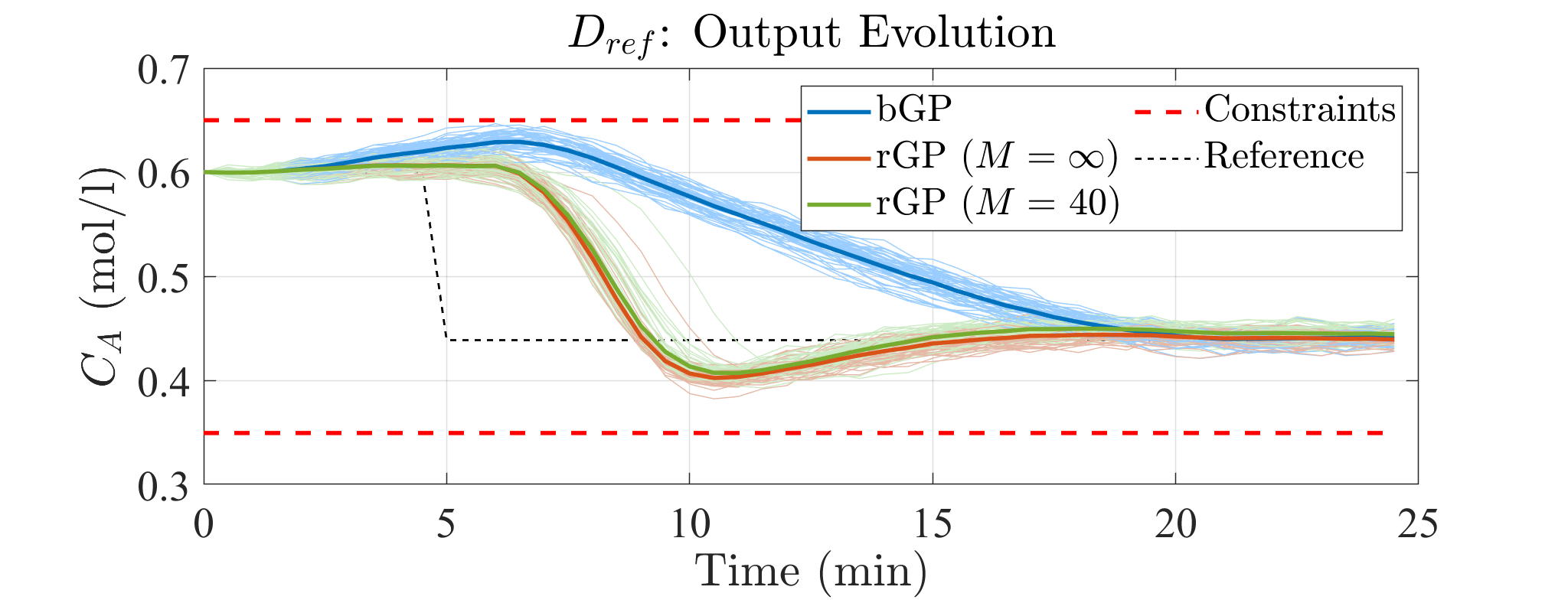}
  \caption{Influence of a limited number of training data points on the
  rGP-MPC with initial training data \Dref.
  }
  \label{fig:limitedTrainPoints}
\end{figure}

Besides the computational cost reduction due to the possibility to work
with smaller training data sets, we also illustrate the computational
reduction due to the recursive update of the Cholesky factor.
In Fig.~\ref{fig:compTime} we continue with the \Dref case, where we
add every incoming point to the training data set and compare the
computation times of the full and the recursive update of the Cholesky
factor. 
The results show that the larger the training data set becomes, the
larger the absolute and relative computational reduction.
At $t = \unit[24]{min}$ the full recomputation of the Cholesky factor
increases significantly.
Investigations point to the reason lying in the generation of the
covariance matrix and the inner workings of Matlab's \texttt{chol}
function to compute the Cholesky decomposition.

\begin{figure}[htpb]
  \centering
  \includegraphics[width=0.6\linewidth]{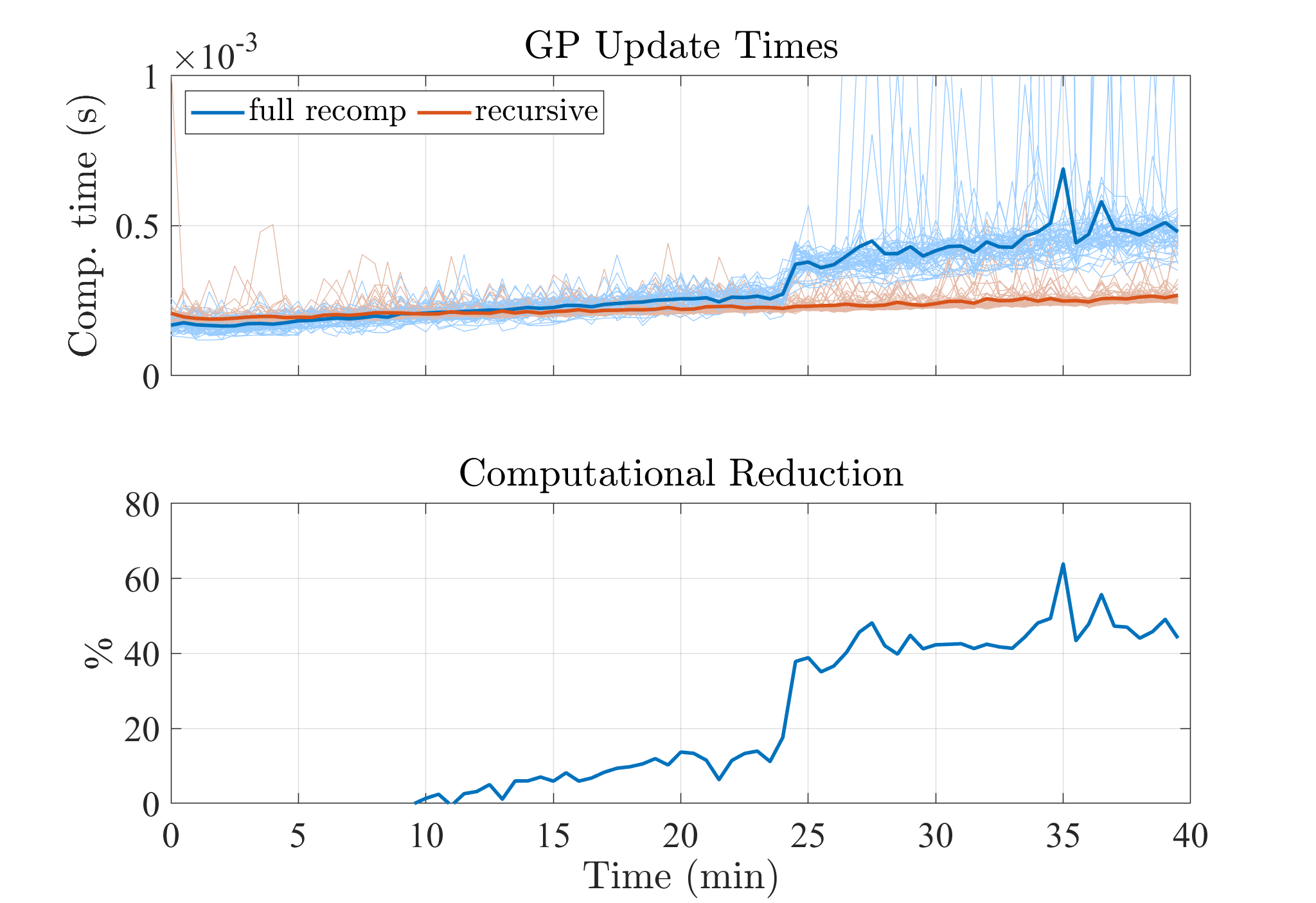}
  \caption{Comparison of computation times of the full recalculation of
  the Cholesky factor and the recursive update.
  The computational reduction that goes along with the recursive update
  increases with the amount of training data points.
  }
  \label{fig:compTime}
\end{figure}

At last we present simulations of the robust feasible set $\Omega^0$, also
denoted region of attraction (ROA), and how it changes for different
maximum prediction errors $\mu$.
We continue with the \Dref case with $\epU = \sigmaU = 0$ such
that every data point is considered as a candidate for inclusion.
Furthermore, $M$ is set to a large value such that no points are removed
from \Dk.
Each initial condition $\x_0 = [y_0 \ y_0 \ y_0]\T$ is simulated
30 times.
Different $\mu$ values are obtained by varying the measurement noise
from $\Sn = 0.003^2$ to $\Sn = 0.012^2$, where $\mu$ is then the largest
error of all simulation runs and time steps.
The result in Fig.~\ref{fig:ROA} yields a clear tendency.
The larger $\mu$, the smaller $\Omega^0$.

\begin{figure}[htpb]
  \centering
  \includegraphics[width=0.6\linewidth]{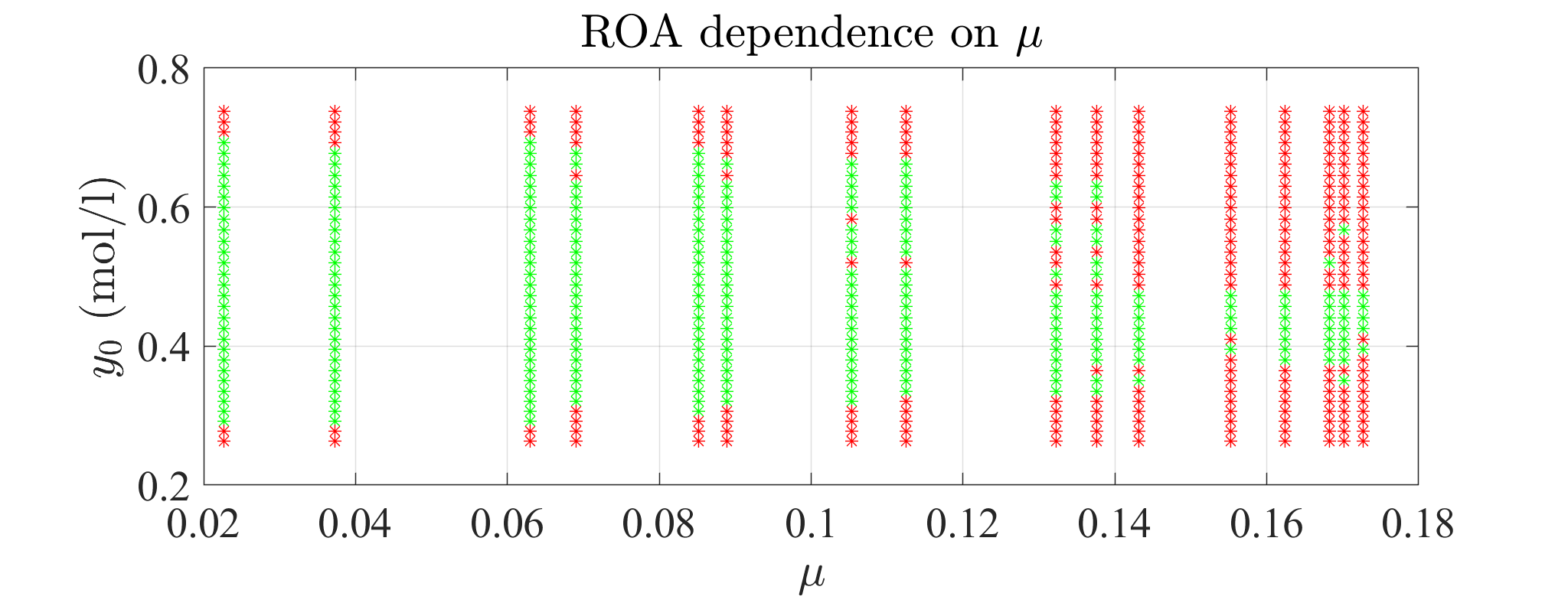}
  \caption{
    Change of the region of attraction $\Omega^0$ for different $\mu$.
    Red stars denote infeasible initial conditions, green stars
    feasible initial conditions.
    An initial condition is marked as infeasible if at least one
    simulation resulted in a constraint violation.
  }
  \label{fig:ROA}
\end{figure}

\section{Conclusion}
\label{sec:conclusion}

In this work, we outlined the use of a \GP-based nonlinear
autoregressive model with exogenous input for prediction in an output
feedback \MPC scheme.
The approach allows for online learning, by means of updating the
training data set, to account for limited a priori process knowledge and
the possibility for adaptation during operation.
To this end, the concept of evolving GPs was adapted together with a
recursive formulation to update the Cholesky decomposition to minimize
computational cost.
It was shown that the resulting \MPC scheme is input-to-state stable
with respect to the prediction error, despite
the time-varying nature of the GP prediction model.
Notably, the theoretic guarantees are not limited to
\GPs.
They are rather valid for all online learning methods that satisfy the
presented conditions.

The approach was verified in simulations, which have shown 
that it is in general possible to start with limited a priori process
knowledge and refining the model during operation.
One important finding is that it is particularly beneficial to start
with a model that captures at least the behavior at the target
reference, which is fortunately an intrinsic necessity for all MPC
schemes that use a terminal region, cost, and controller to
guarantee recursive feasibility and stability.
In the case of fixed hyperparameters during online operation, 
a further consequence is that the hyperparameters should be
optimized offline for a data set that captures the target reference.
Furthermore, the presented formulation yields good
closed-loop performance with few training data
points, thereby efficiently reducing the computational load.
This presents itself as a possible option for very fast processes,
where hyperparameter optimization is not an option but some kind of
online learning is desirable.
Additionally, due to the output feedback scheme, this approach can be
employed for processes, whose state cannot be measured or is difficult
to be estimated.

Future work aims at implementing the presented approach in laboratory
experiments, together with a combination of a deterministic base model
and the \GP prediction model.
From a theoretical point of view, time-varying reference tracking
instead of set-point changes would be interesting to investigate.
For instance, what conditions does the initial training data set has to
satisfy to achieve acceptable tracking results and how to automatically
compute safe thresholds for the data inclusion approach.
Another interesting question to investigate is how the approach
performs for time-varying processes.
A hypothesis would be to combine the squared exponential covariance
function with a non-stationary one to account for time variance in the
process model.

\section*{Acknowledgments}
The authors would like to thank the reviewers for helpful suggestions
and discussions.







\appendix

\section{Recursive Cholesky Factor Update}
\label{sec:recursive_cholesky_factor_update}

According to Osborne\cite{Osborne2010}, the Cholesky factor can be
updated recursively as presented in the following.
Regarding the case of including a new data point, consider the
covariance matrix \KK, represented in block form as
\begin{align*}
  \Matrix{
    \KK_{11}   & \KK_{13} \\
    \KK_{13}\T & \KK_{33}
  }
\end{align*}
and its Cholesky factor
\begin{align*}
  \Matrix{
    \RR_{11}   & \RR_{13} \\
    0          & \RR_{33}
  } \ .
\end{align*}
Now, given an updated covariance matrix
\begin{align*}
  \Matrix{
    \KK_{11}   & \KK_{12}   & \KK_{13} \\
    \KK_{12}\T & \KK_{22}   & \KK_{23} \\
    \KK_{13}\T & \KK_{23}\T & \KK_{33}
  }
\end{align*}
that differs from the previous by insertion of a new row and
column, the updated Cholesky factor
\begin{align*}
  \Matrix{
    \S_{11} & \S_{12} & \S_{13} \\
    0       & \S_{22} & \S_{23} \\
    0       & 0       & \S_{33}
  } \ .
\end{align*}
can be computed via
\begin{align}
\begin{aligned}
  \S_{11} &= \RR_{11}&
  \S_{22} &= \text{chol}\big(\KK_{22} - \S_{12}\T \S_{12} \big) \\
  \S_{12} &= \RR_{11}\T \backslash \KK_{12}&
  \S_{23} &= \S_{22}\T \backslash \big(\KK_{23} - \S_{12}\T \S_{13}\big) \\
  \S_{13} &= \RR_{13}&
  \S_{33} &= \text{chol}\big(\RR_{33}\T \RR_{33} - \S_{23}\T \S_{23}\big) \ .
\end{aligned}
  \label{eq:cholUpdate}
\end{align}

On the other hand, if the current covariance matrix in block form
\begin{align*}
  \Matrix{
    \KK_{11}   & \KK_{12}   & \KK_{13} \\
    \KK_{12}\T & \KK_{22}   & \KK_{23} \\
    \KK_{13}\T & \KK_{23}\T & \KK_{33}
  }
\end{align*}
with Cholesky factor
\begin{align*}
  \Matrix{
    \RR_{11} & \RR_{12} & \RR_{13} \\
    0        & \RR_{22} & \RR_{23} \\
    0        & 0        & \RR_{33}
  } 
\end{align*}
is reduced by one row and column, such that we obtain
\begin{align*}
  \Matrix{
    \KK_{11}   & \KK_{13} \\
    \KK_{13}\T & \KK_{33}
  } \ ,
\end{align*}
the downdated Cholesky factor
\begin{align*}
  \Matrix{
    \S_{11}   & \S_{13} \\
    0         & \S_{33}
  } 
\end{align*}
can be computed via
\begin{align}
\begin{aligned}
  \S_{11} &= \RR_{11} \\
  \S_{13} &= \RR_{13} \\
  \S_{33} &= \text{chol}\big(\RR_{23}\T \RR_{23} + \RR_{33}\T
    \RR_{33}\big) \ .
\end{aligned}
\label{eq:cholDowndate}
\end{align}

\section{Posterior Mean Gradient}
\label{sec:posterior_mean_gradient}

The optimal control problem \eqref{eq:OCP} requires a terminal
cost function, which can be based on a linearized version of
the prediction model in Section~\ref{sec:simulations}.
To this end we require the gradient of the GP posterior mean function
\begin{align*}
  \nabla \postM(\w) = \frac{\partial \postM(\w)}{\partial \w}
    = \left[ \frac{\partial \postM(\w)}{\partial w_1} \ \cdots \
      \frac{\partial \postM(\w)}{\partial w_\Nw} \right]\T
\end{align*}
w.r.t. to its regressor $\w = [w_1, \ldots, w_\Nw]$, where we omit 
the dependence on the training data \D for the sake of brevity.

Assuming a constant prior mean in \eqref{eq:postM} we obtain
\begin{align*}
  \nabla \postM(\w) = \frac{\partial k(\w,\wb)}{\partial \w} \KKi (\zb -
    m(\w))
\end{align*}
with
\begin{align*}
  \frac{\partial k(\w,\wb)}{\partial \w} = 
    \left[ \frac{\partial k(\w,\w_1)}{\partial \w} \cdots \frac{\partial
    k(\w,\w_n)}{\partial \w} \right] \ ,
\end{align*}
where $\w_1, \ldots, \w_n$ is the corresponding regressor of each of the
$n$ measured training data points in \D.

For the covariance function \eqref{eq:covSEard}, we obtain
\begin{align*}
  \frac{\partial k(\w,\w')}{\partial \w} = k^*(\w,\w') \Lambda (\w'-\w)
  \ ,
\end{align*}
where $k^*(\w,\w')$ is \eqref{eq:covSEard} without the noise term, \ie
$\Sn = 0$.


\bibliography{bibliography}%

\begin{table}[htpb]
  \centering
  \caption{CSTR Parameters}
  \label{tab:CSTRparameters}
    \begin{tabular}{lll}
    \toprule
      Param. & Explanation & Value \\
    \midrule
      $q_0$ & Reactive input flow & \unit[10]{l/min} \\
      $V$   & Liquid volume in the tank & \unit[150]{l} \\
      $k_0$ & Frequency constant & \unit[$6\cdot10^{10}$]{1/min} \\
      $E/R$ & Arrhenius constant & \unit[9750]{K} \\
      $\Delta H_\text r$ & Reaction enthalpy & \unit[10000]{J/mol} \\
      $UA$  & Heat transfer coefficient & \unit[70000]{J/(min K)} \\
      $\rho$ & Density & \unit[1100]{g/l} \\
      $C_\text p$ & Specific heat & \unit[0.3]{J/(g K)} \\
      $\tau$ & Time constant & \unit[1.5]{min} \\
      $C_{A \text f}$ & $C_A$ in the input flow & \unit[1]{mol/l} \\
      $T_\text f$ & Input flow temperature & \unit[370]{K} \\
    \bottomrule
    \end{tabular}
\end{table}

\begin{table}[htpb]
  \centering
  \caption{Hyperparameters}
  \label{tab:hyperparameters}
    \begin{tabular}{lcccccc}
    \toprule
              & $c$ & $l_1$ & $l_2$ & $l_3$ & $l_4$ & \Sf \\
    \midrule
      \DO     & 0.64 & 0.07 & 0.29 & 0.14 & 9.93 & 0.06 \\
      \Dref   & 0.36 & 0.20 & 11.7 & 0.64 & 5.07 & 0.13 \\
      \Dcomb  & 0.43 & 0.42 & 2.09 & 1.01 & 2.83 & 0.26 \\
    \bottomrule
    \end{tabular}
\end{table}

\begin{table}[htpb]
  \centering
  \caption{MPC Performance computed by \eqref{eq:performance}.}
  \label{tab:performance}
    \begin{tabular}{lcccccc}
    \toprule
              & $\DO $ & \Dref & \Dcomb  \\
    \midrule
      oMPC    & 59.5   & 59.5  & 59.5 \\
      bGP-MPC & 71.3   & 95.3  & 66.2 \\ 
      rGP-MPC & 64.5   & 63.6  & 66.7 \\ 
    \bottomrule
    \end{tabular}
\end{table}

\begin{algorithm}
\caption{Recursive Guassian Process Model Predictive Control}
  \label{alg1}
\begin{algorithmic}
\item \textbf{MPC Parameters:} Prediction horizon $N$, stage cost
  $\ell(\cdot)$ with respective parameters, hard input constraint set
  \U, output constraint set \Y. 
\item \textbf{rGP Parameters:} Prior mean $m(\cdot)$, covariance
  function $\cov(\cdot,\cdot)$, initial hyperparameters \thetaa,
  thresholds \epU and \sigmaU, maximum number of training points $M$.
\item 
\item \textbf{Initialization} 
\item Training data set \D.
\item Optimize hyperparameters \thetaa \eqref{eq:loglike} with initial
  data set \D.
\item Initialize GP posterior mean function $\postM(\w)$ with covariance
  matrix \KK, Cholesky factor \RR, and \alphaa
  (Sec.~\ref{sec:cholesky_decomposition}).
\item Compute GP posterior mean gradient $\nabla \postM(\w)$
  (Sec.~\ref{sec:posterior_mean_gradient}).
\item Compute linear GP model at \xref
  (Sec.~\ref{sec:terminal_ingredients}).
\item Compute terminal cost function $\Vf(\cdot)$
  (Sec.~\ref{sec:terminal_ingredients}).
\item 
\item \textbf{Recursion}
  \For{each time step $k$}
    \State Solve optimal control problem \eqref{eq:OCP} for initial
    condition \xk and obtain optimal input sequence $\uS^*$.
    \State Apply first element $\uk = \kMPC(\xk|\Dk) = \hu^*_{k|k}$.
    \State Obtain new output \ykp.
    \State Construct new GP data point $(\wk,\ykp)$ with $\wk =
      (\xk,\uk)$.
    \State \textbf{Update GP:}
    \State Compute $\hykp = \postM(\wk|\Dk)$ and $\postV = \postV(\wk|\Dk)$.
    \If{$\abs{\ykp - \hykp} > \epU$ OR $\postV > \sigmaU$}
      \State $\Dskp = \Dk \cup (\wk,\ykp)$.
      \State Using \Dskp, compute \KK' and \RR' via \eqref{eq:cholUpdate}.
      \If{number of training points $> M$}
        \State Remove oldest data point and downdate \KK' and \RR' via
        \eqref{eq:cholDowndate}.
      \EndIf
      \State Compute \alphaa' via \eqref{eq:alphaa}.
      \If{$V_N^*\big( \xk | \Dskp \big) \leq V_N^*\big( \xk | \Dk \big)$}
        \State $\Dkp = \Dskp$
        \State Make \KK', \RR', and \alphaa' effective.
        \Else
        \State $\Dkp = \Dk$
        \State Reverse \KK', \RR', and \alphaa'.
      \EndIf
      \Else
      \State $\Dkp = \Dk$
    \EndIf
  \EndFor
\end{algorithmic}
\end{algorithm}




\end{document}